\documentclass[floats,12pt,onecolumn]{article}
\usepackage{amsfonts}
\usepackage{amsmath}
\usepackage{calc}
\usepackage{subfigure}
\usepackage{graphicx} 
\usepackage{enumerate}
\usepackage{amssymb}
\usepackage{amsthm}
\newcommand{\appref}[1]{\hyperref[#1]{{Appendix~\ref*{#1}}}}
\newcommand{\be}{\begin{eqnarray} \begin{aligned}}
\newcommand{\ee}{\end{aligned} \end{eqnarray} }
\newcommand{\benn}{\begin{eqnarray*} \begin{aligned}}
\newcommand{\eenn}{\end{aligned} \end{eqnarray*}}
\newcommand{\cancel}[1]{} 
\newcommand*{\textfrac}[2]{{{#1}/{#2}}}

\newcommand*{\cB}{\mathcal{B}}

\newcommand*{\cE}{\mathcal{E}}

\newcommand*{\cH}{\mathcal{H}}

\newcommand*{\cL}{\mathcal{L}}
\newcommand*{\cM}{\mathcal{M}}

\newcommand*{\cR}{\mathcal{R}}

\newcommand*{\cW}{\mathcal{W}}

\newcommand*{\tr}{\mathop{\mathrm{tr}}\nolimits}

\newcommand*{\myspan}{\mathrm{span}}

\newcommand{\bc}{\begin{center}}
\newcommand{\ec}{\end{center}}

\newtheorem{theorem}{Theorem}[section]
\newtheorem{lemma}[theorem]{Lemma}

\newtheorem{corollary}[theorem]{Corollary}

\usepackage{amsfonts}

\def\01{\{0,1\}}

\newcommand{\ket}[1]{|#1\rangle}
\newcommand{\bra}[1]{\langle#1|}

\newcommand{\proj}[1]{|#1\rangle\langle#1|}

\newcommand{\comment}[1]{}
\newcommand{\diag}{\mbox{diag}}

\begin{document}
\newcommand*{\spec}{\mathsf{spec}}
\newcommand*{\sspec}{\mathsf{Sspec}}
\newcommand*{\photonnumber}{{\bf N}}
\newcommand*{\Mat}{\mathsf{Mat}}
\title{The conditional entropy power inequality\\
for Gaussian quantum states} 

\author{Robert Koenig}
\date{\today}
\maketitle
\begin{abstract}
We propose a generalization of the quantum entropy power inequality involving conditional entropies.  For the special case of Gaussian states, we give a proof based on  perturbation theory for symplectic spectra. We discuss some implications for entanglement-assisted classical communication over  additive bosonic noise channels. 
\end{abstract}

\section{Classical and quantum entropy-power inequalities}
The entropy power inequality, proposed by Shannon~\cite{Shannon48} and later established with increasing rigor by Stam~\cite{Stam59} and Blachman~\cite{Blachman65},
has become a fundamental tool in classical information theory. Shannon's original application of the entropy power inequality is  a lower bound on the capacity of an additive (but potentially non-Gaussian) noise channel~\cite[Theorem 18]{Shannon48}. However, the usefulness of the entropy power inequality is especially  evident in  multi-terminal information theory. Among the most well-known applications are the characterization of the    Gaussian broadcast channel by Bergman~\cite{Bergmans}, 
 the Gaussian two-description problem by Ozarow~\cite{Ozarow} and the quadratic Gaussian CEO problem by Oohama~\cite{Oohama}. In these multi-user settings, Fano's inequality by itself is insufficient to characterize different tradeoffs. A more recent application proposed by Liu and Viswanath~\cite{liuvish} uses entropy power inequalities to solve certain optimization problems.

The entropy power inequality lower bounds the differential entropy  of the con\-volution  of two  independent random variables~$X$,~$Y$ taking values in~$\mathbb{R}^n$. Its  covariance-preserving version states that
\begin{align}
H(\sqrt{\lambda} X+\sqrt{1-\lambda}Y)\geq \lambda H(X)+(1-\lambda) H(Y)\qquad\textrm{ for all }0\leq \lambda\leq 1\ .\label{eq:covariancepreservingentropypower}
\end{align}
Eq.~\eqref{eq:covariancepreservingentropypower} can be shown~\cite{Lieb78,VerduGuo06} to be equivalent to
the more commonly used statement 
\begin{align*}
e^{2H(X+Y)/n}\geq e^{2H(X)/n}+e^{2H(Y)/n}\ .
\end{align*}  The latter explains the terminology as $e^{2H(X)/n}$ is the power, i.e., variance of a Gaussian random variable with identical entropy as~$H(X)$. Inequalities such as~\eqref{eq:covariancepreservingentropypower} are closely related to Log-Sobolev inequalities (see e.g.,~\cite{Toscani}) 
as well as Brunn-Minkowski-type inequalities~\cite{CostaThomas84}.  Generalizations to free probability~\cite{noncommutativeprobabilitytheory} and quantum states have been considered.

\begin{center}
\begin{figure}
\begin{center}
\includegraphics{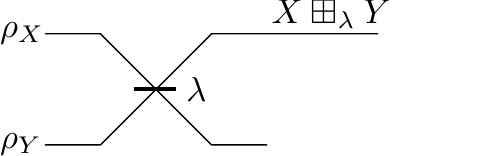}
\caption{\label{fig:epi}The setting of the (unconditional) quantum entropy power inequality: two independent sets of modes (i.e., a product state~$\rho_X\otimes\rho_Y$) combine at a beam-splitter of transmissivity~$\lambda\in [0,1]$.  }
\end{center}
\end{figure}
\end{center}

Among known quantum generalizations is the photon-number-inequality conjectured in~\cite{Guhaetal07} (and proved for special cases such as Gaussian states in~\cite{Guhathesis}), and a version of~\eqref{eq:covariancepreservingentropypower} involving von Neumann entropies $S(\rho)=-\tr(\rho\log\rho)$ instead of differential (Shannon) entropies~\cite{KoeGrae}. Formally, this statement is obtained
by substituting states~$\rho_X,\rho_Y$ of~$n$ bosonic modes for
the random variables $(X,Y)$, and using a beamsplitter to process the product state~$\rho_X\otimes\rho_Y$, see Fig.~\ref{fig:epi}. 
 This results in an output denoted~$\rho_{X\boxplus_\lambda Y}$ on one of the output arms (see below for a precise definition), and the corresponding quantum generalization states that
\begin{align}
S(X\boxplus_\lambda Y)\geq \lambda S(X)+(1-\lambda)S(Y)\qquad\textrm{ for all }0\leq \lambda\leq 1\ .\label{eq:entropypowerquantumunconditional}
\end{align}
As discussed in~\cite{KoeGraeNat}, Eq.~\eqref{eq:entropypowerquantumunconditional} provides strong upper limits on classical communication over additive thermal noise channels. Indeed, for suitable parameter regimes, the resulting upper bounds  on the classical capacity are close to the Holevo-Schumacher-Westmoreland lower bound  achievable by coherent states. This limits the degree of  potential additivity violations, implying that coding strategies using simple product states are close to optimal for thermal noise channels.

\section{Conditional entropy-power inequalities}
Most applications of the classical entropy power inequality make use of a version for conditional entropies $H(X|Z)=H(XZ)-H(Z)$. It is easy to see that, if $(X,Y)$ are conditionally independent given~$Z$, then
\begin{align}
H(\sqrt{\lambda} X+\sqrt{1-\lambda}Y|Z)\geq \lambda H(X|Z)+(1-\lambda) H(Y|Z)\qquad\textrm{ for all }0\leq \lambda\leq 1\ .\label{eq:conditionalentropypowerinequalityclassical}
\end{align}
Indeed,~\eqref{eq:conditionalentropypowerinequalityclassical} is an immediate consequence of~\eqref{eq:covariancepreservingentropypower} and the fact that
the conditional entropy~$H(X|Z)=\sum_z P_Z(z) H(X|Z=z)$ is the average of the entropies~$H(X|Z=z)$ of the conditonal distributions~$\{P_{X|Z=z}\}_z$.

It is natural to ask whether a generalization of~\eqref{eq:conditionalentropypowerinequalityclassical} is true for states~$\rho_{XYZ}$ for which the $n$-mode systems $X$ and $Y$ are conditionally independent given the quantum system~$Z$, i.e., whether
\begin{align}
S(X\boxplus_\lambda Y|Z)\geq \lambda S(X|Z)+(1-\lambda)S(Y|Z)\ .\label{eq:qconditionalentropyinequality}
\end{align}
We will argue that~\eqref{eq:qconditionalentropyinequality} is useful to estimate entanglement-assisted capacities of additive noise channels. The inequality may have additional applications in multi-user quantum information theory.

Establishing an inequality of the form~\eqref{eq:qconditionalentropyinequality} appears to be non-trivial because we cannot simply condition on the quantum system~$Z$ (unless, of course, it is purely classical). However, the following simplification is immediate: it suffices to establishes 
an inequality of the form
\begin{align}
S(X\boxplus_\lambda Y|Z_1Z_2)\geq \lambda S(X|Z_1)+(1-\lambda)S(Y|Z_2)\
\textrm{ for all product states } \rho_{XZ_1}\otimes\rho_{YZ_2}\ .
 \label{eq:conditionalentropypowerinequalitymain}
\end{align}
 This is because any conditionally independent state~$\rho_{XYZ}$ has the Markov form~\cite{HayJoPeWi04}
\begin{align*}
\rho_{XYZ}&=\bigoplus_j p_j \rho_{XZ^{(j)}_1}\otimes\rho_{YZ^{(j)}_2}\ 
\end{align*}
and the von Neumann entropy satisfies~$S(\bigoplus_j p_j \rho_j)=\sum_j p_j S(\rho_j)+H(p)$ on direct sums, where $H(p)$ is the  Shannon entropy of the distribution~$\{p_j\}_j$.

Here we prove inequality~\eqref{eq:conditionalentropypowerinequalitymain} for all pairs of Gaussian states~$\rho_{XZ_1}$ and~$\rho_{YZ_2}$. We conjecture that~Eq.~\eqref{eq:qconditionalentropyinequality} holds in general for arbitrary conditionally independent states~$\rho_{X_1X_2Z}$. If this is the case, the implications for entanglement-assisted capacities discussed here extend to all additive (but not necesssarily Gaussian) channels. 

It may be possible to find a proof  of~Eq.~\eqref{eq:qconditionalentropyinequality} using a similar strategy as we employ in the Gaussian case. However, this will require a novel analysis of the evolution of conditional entropies under a certain Markovian evolution. In the Gaussian case, this is based on perturbation theory for symplectic eigenvalues, as we explain below.

\section{Implications for entanglement-assisted  communication}
Quantum communication channels are characterized by different capacities depending on the additional auxiliary resources available, as well as whether or not the communicated information is classical or quantum. Here we focus on entanglement-assisted classical capacities which are arguably best understood. Consider a point-to-point scenario where a sender~$A$ tries to communicate to a receiver~$C$ over a channel~$\cE:\cB(A)\rightarrow\cB(C)$.  The entanglement-assisted classical capacity~$C_E(\cE)$ is defined as the maximal rate (in bits/channel use) at which classical bits can be transmitted reliably if the sender and receiver share an unlimited amount of prior entanglement.

In sharp contrast to the unassisted classical~\cite{Hastings} or the quantum capacity~\cite{SmithYard08}, the  quantity~$C_E(\cE)$ is additive~\cite{adamicerf}.  In particular, it has the single-letter expression~$C_E(\cE)=\sup_\rho I(\cE,\rho)$
in terms of the quantum mutual information
\begin{align*}
I(\cE,\rho)=S(\rho)+S(\cE(\rho))-S((\cE\otimes I_{A'})(\Psi_{AA'}))=:I(A':C)_{(\cE\otimes I_{A'})(\Psi_{AA'})}\ ,
\end{align*} where~$\Psi_{AA'}$ is a purification of the input density operator~$\rho$. (This statement is a generalization~\cite{Bennettetal02,Holevoentanglementassisted} of the Holevo-Schumacher-Westmoreland theorem.)   The quantity~$I(\cE,\rho)$ has a number of nice properties: it is positive and  concave with respect to the input state~$\rho$.

 For channels involving infinite-dimensional state spaces, such capacity results can be
adapted by certain limiting procedures -- see~\cite{Holevo03} for a rigorous analysis of the entanglement-assisted capacity. It is necessary to constrain the inputs to obtain meaningful results: one is led to consider the quantity~$C_E(\cE,\photonnumber)$ obtained by maximizing over all input states~$\rho$ with mean photon number  upper bounded by~$\photonnumber$, i.e., 
\begin{align}
C_E(\cE,\photonnumber)=
\sup_{\rho: \langle a^\dagger a\rangle_\rho\leq \photonnumber} I(\cE,\rho)\ .\label{eq:generalizationoverallinputstates}
\end{align}
\begin{center}
\begin{figure}
\begin{center}
\includegraphics{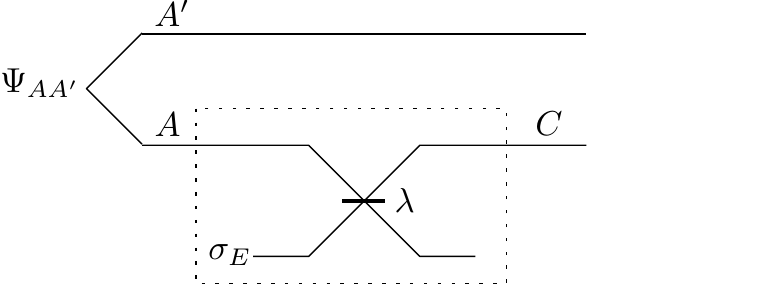}
\caption{\label{fig:pointtopoint}The additive noise channel~$\cE_{\lambda,\sigma_E}$  with a transmissivity~$\lambda$-beamsplitter and environment in the state~$\sigma_E$ is schematically illustrated by the dotted box. Its entanglement-assisted capacity~$C_E(\cE_{\lambda,\sigma_E},\photonnumber)$ is obtained by  maximizing the mutual information $I(A':C)$ over all input states~$\rho_{A}$ (with purification~$\Psi_{AA'}$) subject to the mean photon number constraint~$\tr(a^\dagger a\rho)\leq \photonnumber$. }
\end{center}
\end{figure}
\end{center}
 
We are interested in estimating~\eqref{eq:generalizationoverallinputstates} when~$\cE=\cE_{\lambda,\sigma_E}$ is an additive noise channel. Such a channel is characterized by the transmissivity~$\lambda\in [0,1]$  and a state~$\sigma_E$ of the environment, see Figure~\ref{fig:pointtopoint}. Both the input as well as the environment of the channel $\cE_{\lambda,\sigma_E}$ consist of $n$~bosonic modes; the two interact with a beamsplitter of transmissivity~$\lambda$. The output of the channel is a set of $n$~modes (see Section~\ref{sec:beamsplittersproductstates} for a precise definition of these expressions).  We will focus on $n=1$ (this often being sufficient because of additivity), although the entropy power inequality also applies for~$n>1$. 

The entanglement-assisted capacity of~$\cE_{\lambda,\sigma_E}$ can easily be upper bounded by twice the maximum output entropy, i.e., \footnote{The proof of~\eqref{eq:naiveupperboundentanglementassisted} follows immediately from (see Figure~\ref{fig:pointtopoint})
\begin{align*}
I(A':C)=S(C)-S(C|A')=S(C)+S(C|E'D)\leq 2S(C)\ .
\end{align*}
Here $E'$ purifies the environment, and $D$ is the second beam-splitter output. 
The second identity uses the purity of the overall state on $A'CE'D$ and the inequality is subadditivity of the entropy.}
\begin{align}
C_E(\cE_{\lambda,\sigma_E},\photonnumber)\leq 2\max_{\rho: \langle a^\dagger a\rangle_{\rho}\leq \photonnumber} S(C)\leq 2g(\photonnumber_{\max})\ .\label{eq:naiveupperboundentanglementassisted}
\end{align}
Here the second inequality uses the fact~\cite{Wolfetal06} that Gaussian states maximize entropy under a given photon number constraint:  $g(x)=(x+1)\log(x+1)-x\log x$ is the entropy (in nats) of a Gaussian state with mean photon number~$x$, $\photonnumber_{\max}=\lambda \photonnumber+(1-\lambda) \photonnumber_E$ is the maximal mean photon number at the output, and $\photonnumber_E=\langle a^\dagger a\rangle_{\sigma_E}$ is the mean photon number of the environment. In the limit~$\lambda\rightarrow 1$ of perfect transmission, the rhs.~of~\eqref{eq:naiveupperboundentanglementassisted}  becomes twice the entropy of the input as achievable by dense coding~\cite{densecoding}.
\noindent The following Corollary to our conditional entropy power inequality (Theorem~\ref{thm:entropypowerinequalitygaussian} below) improves on the upper bound~\eqref{eq:naiveupperboundentanglementassisted}.

\begin{corollary}\label{cor:entanglementassistedcapacities} Let $\cE_{\lambda,\sigma_E}$ be the additive noise channel with transmissivity~$\lambda$ and environment in a state~$\sigma_E$ with mean photon number~$\photonnumber_E$. If~$\sigma_E$ is Gaussian, then the entanglement-assisted capacity satisfies
\begin{align}
C_E(\cE_{\lambda,\sigma_E},\photonnumber)\leq g(\photonnumber_{\max})+\lambda g(\photonnumber)-(1-\lambda)S(E)\ ,\label{eq:boundsentanglementassistedcapacities}
\end{align}
where $\photonnumber_{\max}=\lambda \photonnumber+(1-\lambda) \photonnumber_E$. Moreover, if conjecture~\eqref{eq:qconditionalentropyinequality} holds for all states, then the bound~\eqref{eq:boundsentanglementassistedcapacities} holds even in the case where~$\sigma_E$ is not Gaussian. 
\end{corollary}
We emphasize that  Corollary~\ref{cor:entanglementassistedcapacities} is of interest mainly in cases where the channel is not completely characterized as e.g., in quantum cryptography.  Under conjecture~\eqref{eq:qconditionalentropyinequality}, it gives a universal upper bound independent of the detailed structure of the environment's state~$\sigma_E$.

\begin{proof}
Assume that~$\sigma_E$ is Gaussian such that~$\cE_{\lambda,\sigma_E}$
is a Gaussian operation. Then the optimization~\eqref{eq:generalizationoverallinputstates} can be restricted to Gaussian states~\cite{HW01}. Consider an arbitrary Gaussian~$\rho_A$, and let~$\Psi_{AA'}$ be a Gaussian purification. Then
\begin{align*}
I(A':C)=S(C)-S(C|A')\leq S(C)-(\lambda  S(A|A')+(1-\lambda) S(E))
\end{align*}
by the conditional entropy power inequality for Gaussian states. The claim then follows from the maximum entropy principle~\cite{Wolfetal06} (i.e., the fact that Gaussian states maximize entropy under a constraint on the second moments) because $S(A|A')=-S(A)$ for a pure state~$\Psi_{AA'}$.

The case of a non-Gaussian state~$\sigma_E$ follows in a similar manner (under conjecture~\eqref{eq:qconditionalentropyinequality}). Here we cannot  restrict the optimization~\eqref{eq:generalizationoverallinputstates} to Gaussian states.
\end{proof}

Figure~\ref{fig:squeezedstateexample} shows a comparison of this bound with the known capacity of the thermal noise channel.

\begin{figure}[ht!]
     \begin{center}
        \subfigure[{\protect $\photonnumber_E=1/2$, $\photonnumber=10$ and $\lambda\in [0,1]$}]{ 
            \label{fig:first}
            \includegraphics[width=0.5\textwidth]{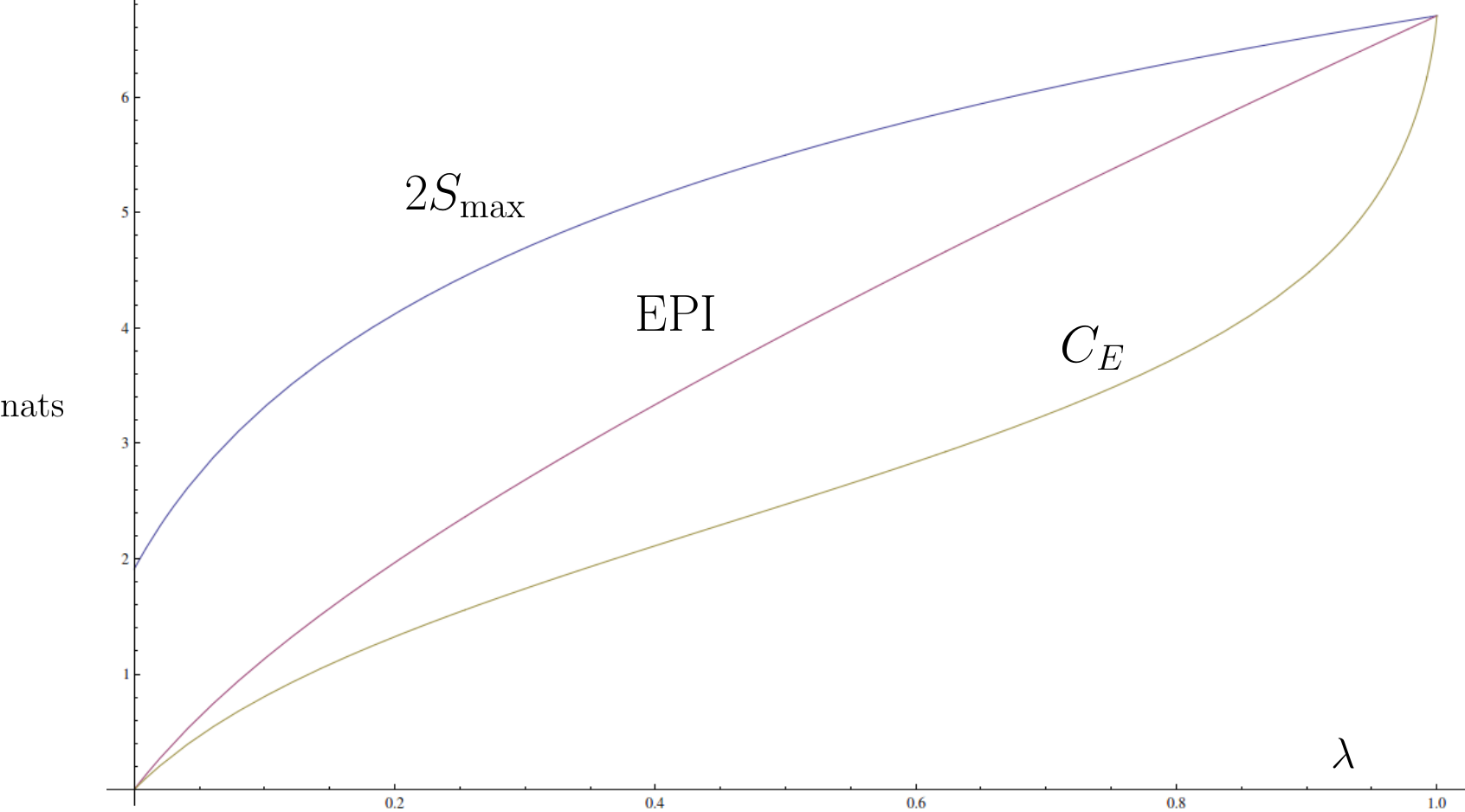}}%
        \subfigure[{\protect $\photonnumber_E=1/2$, $\lambda=1/4$ and  $\photonnumber\in [0,4]$}]{%
            \label{fig:fourth}
            \includegraphics[width=0.5\textwidth]{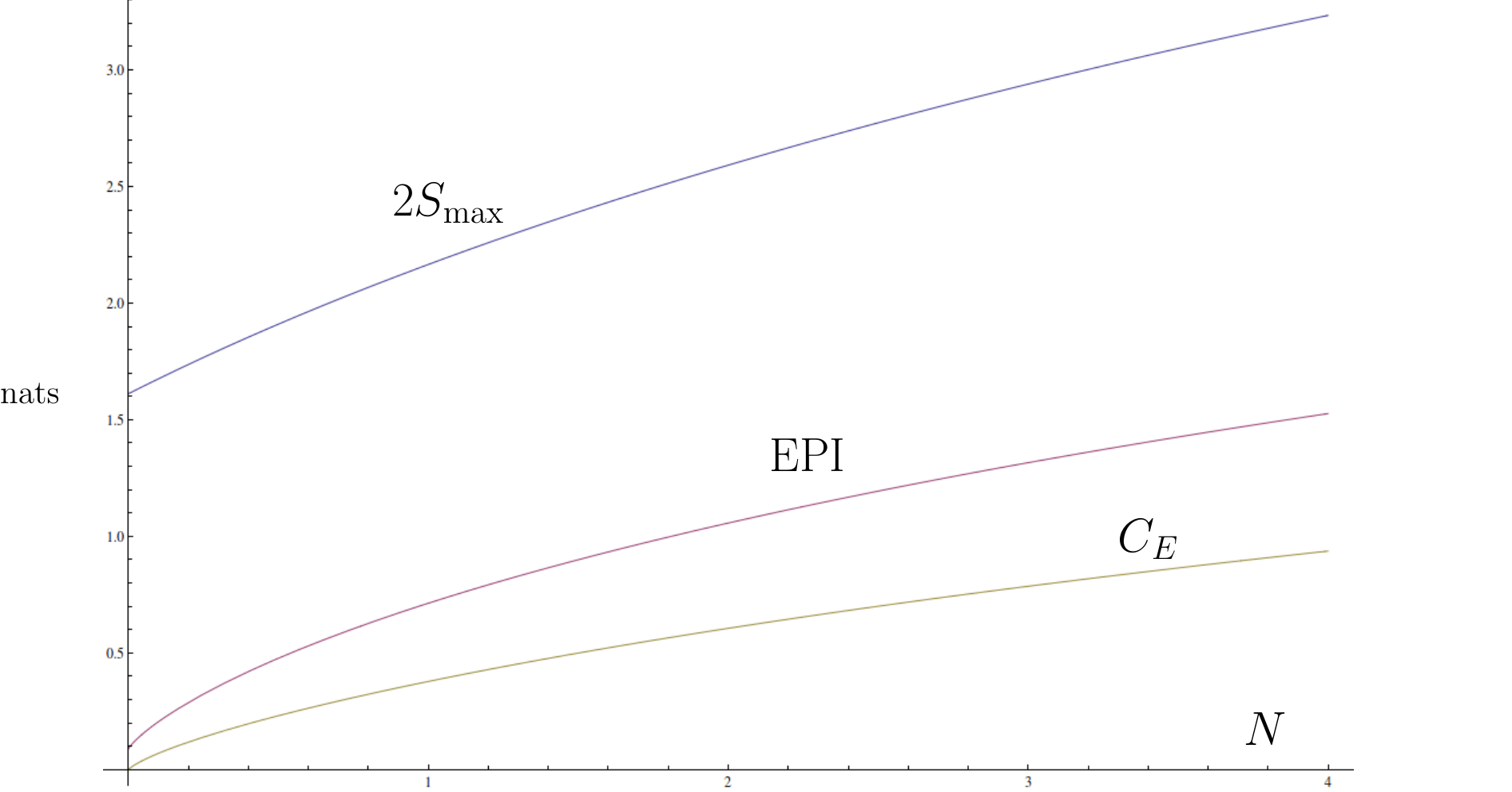}
        }\\ 

        \subfigure[{\protect $\photonnumber_E=0.005$, $\photonnumber=10$ and~$\lambda\in [0,1]$}]{%
            \label{fig:third}
            \includegraphics[width=0.5\textwidth]{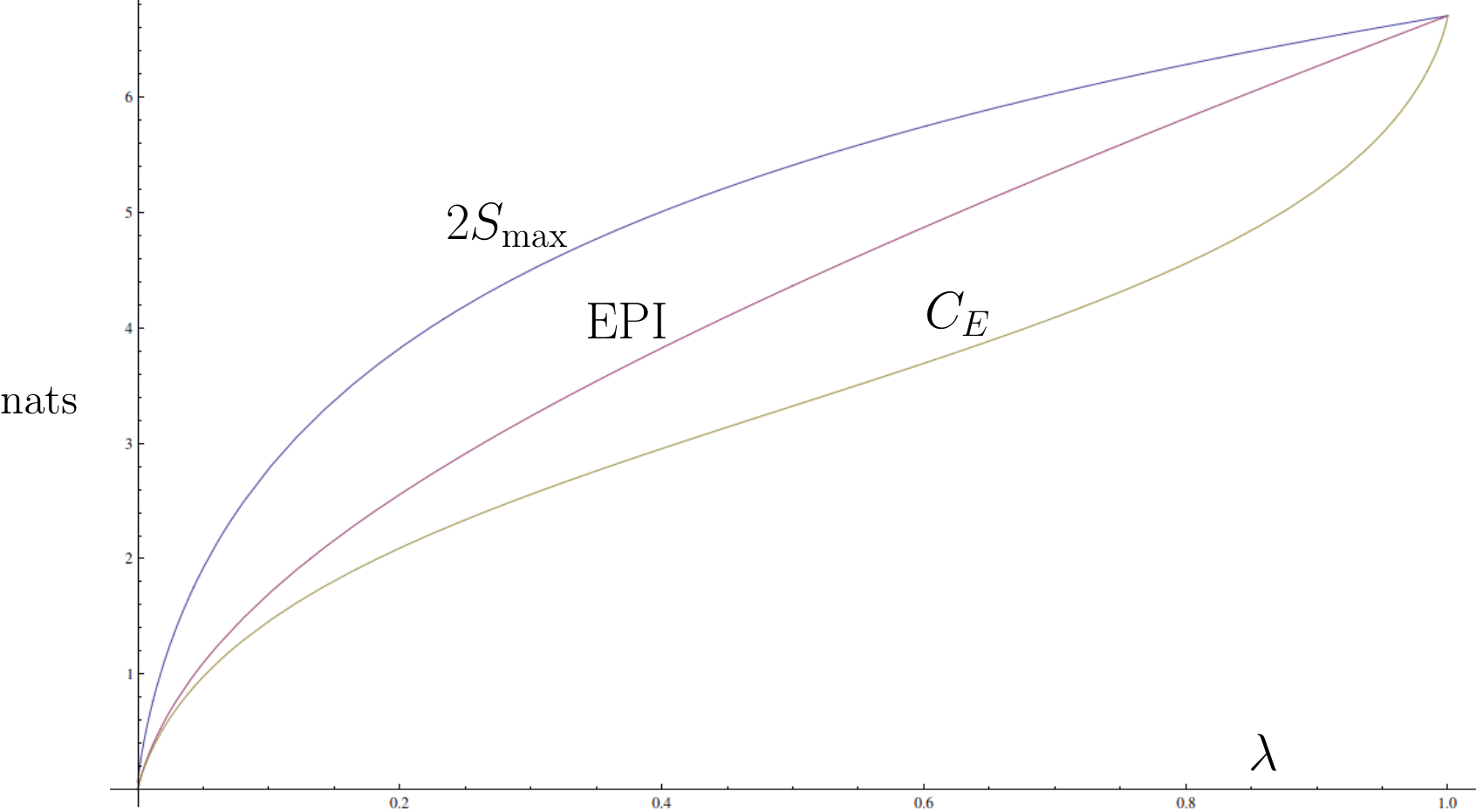}
        }%
        \subfigure[{\protect $\photonnumber_E=0.005$, $\lambda=1/4$ and  $\photonnumber\in [0,4]$}]{%
           \label{fig:second}
           \includegraphics[width=0.5\textwidth]{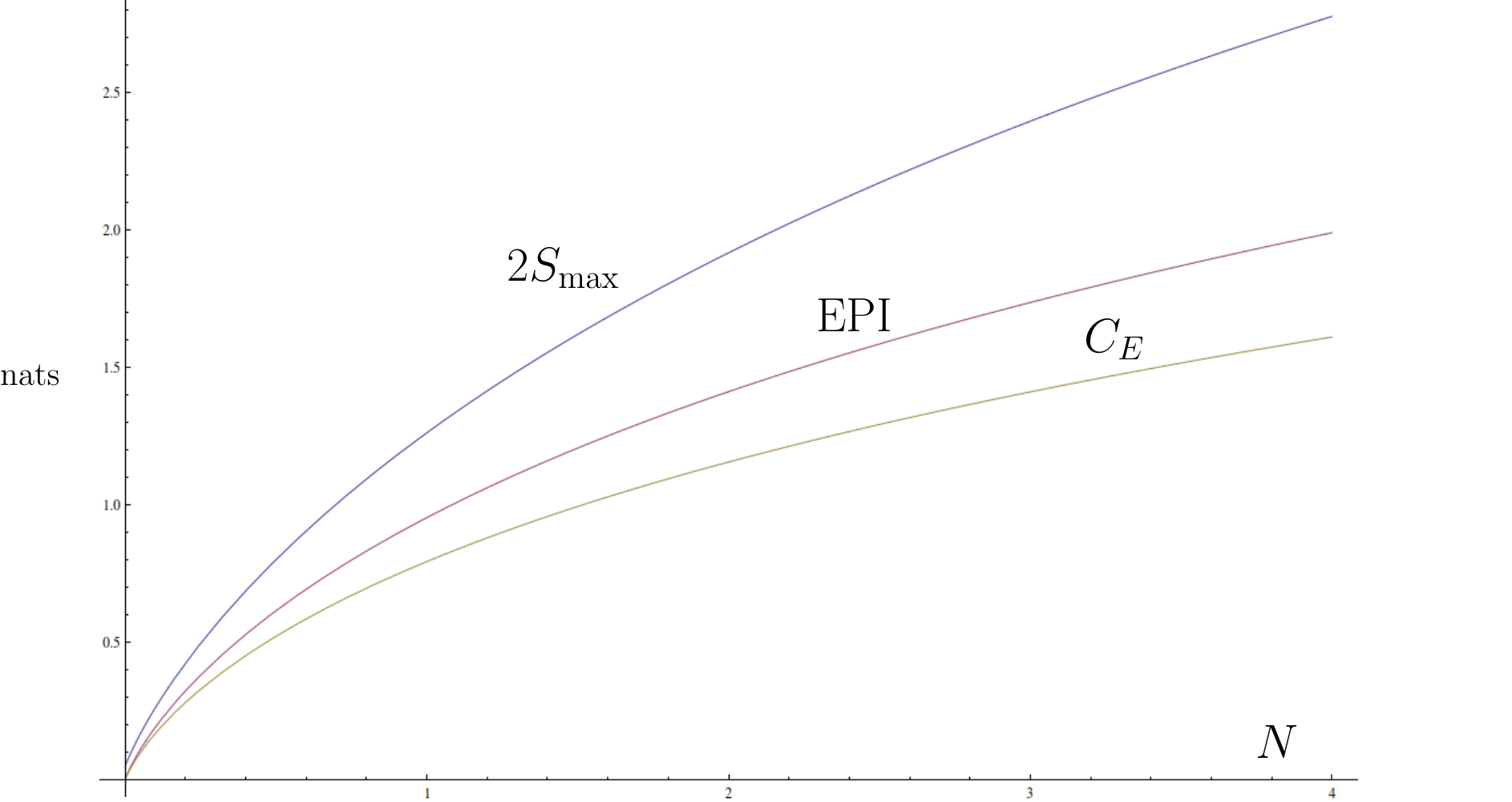}
        }
    \end{center}
\caption{Here we compare
the upper bound of Corollary~\ref{cor:entanglementassistedcapacities} with the known capacity
of the thermal noise channel (with environment photon number~$\photonnumber_E$). The latter is given by~\cite{HW01} the expression $C_E(\cE_{\lambda,\sigma_E},\photonnumber)=g(\photonnumber)+g(\photonnumber_{\max})-g((D+\photonnumber_{\max}-\photonnumber-1)/2)-g((D-\photonnumber_{\max}+\photonnumber-1)/2)$, where $D=\sqrt{(\photonnumber+\photonnumber_{\max}+1)^2-4\lambda \photonnumber(\photonnumber+1)}$.  We are interested in the case where~$\sigma_E$ is close to the vacuum state  corresponding to a pure loss channel. We plot the  capacity $C_E$, the 
EPI upper bound~\eqref{eq:boundsentanglementassistedcapacities}
and the maximum entropy upper bound~\eqref{eq:naiveupperboundentanglementassisted}. While in this case, the exact value of the capacity is known, these figures illustrate how the entropy power inequality improves over the trivial bound. Its insensitivity to the exact form of the environment's state may be useful in certain applications. In contrast, the expression for $C_E$ depends on the fact that $\sigma_E$ is a thermal Gaussian state, that is, the gauge-invariance of~$\cE_{\lambda,\sigma_E}$. \label{fig:squeezedstateexample}}
\end{figure}
Note that if the state~$\sigma_E$ is Gaussian, then the channel~$\cE_{\lambda,\sigma_E}$ is a Gaussian operation.  Various capacities of such channels have been studied in detail~\cite{HW01} (see~\cite{eisertwolf} for a recent review). In particular, the 
entanglement-assisted capacity~$C_E(\cE_{\lambda,\sigma_E},\photonnumber)$ is the result of a convex optimization problem: it is obtained by maximizing the mutual information over the (convex) set of covariance matrices~$M$ associated with Gaussian states~$\rho$ satisfying the photon number constraint\footnote{This follows from the fact that  the maximum mutual information in~\eqref{eq:generalizationoverallinputstates} is achieved on a Gaussian state~$\rho$ and the concavity of~$I(\cE,\rho)$ in~$\rho$, see~\cite{HW01}.}. In principle, this can be addressed using efficient numerical algorithms. Furthermore,   in special cases, it is possible to give an explicit expression for~$C_E(\cE_{\lambda,\sigma_E},\photonnumber)$: Holevo and Werner~\cite{HW01} computed the entanglement-assisted capacity of the attenuation/amplification channel with classical noise (the most general one-mode channel not involving squeezing). Similarly, the entanglement-assisted capacity of the broadband lossy channel was discussed in detail in~\cite{ebcc03}.  The crucial feature enabling these calculations is gauge-invariance of~$\cE_{\lambda,\sigma_E}$. Because the entanglement-assisted capacity is additive, this is equivalent to the state~$\sigma_E$ being thermal with respect to the mode operators defining the beam-splitter. In contrast, Corollary~\ref{cor:entanglementassistedcapacities} gives a bound which is applicable to all Gaussian additive channels without further restrictions.

Recall that the entropy power inequality~\cite{KoeGraeNat} provides an additive upper bound on the classical capacity of thermal noise channels. Because the entanglement-assisted capacity is additive, the conditional entropy power inequality plays a somewhat different role in Corollary~\ref{cor:entanglementassistedcapacities}: it substitutes an optimization problem by a bound depending on simple universal parameters (i.e., the entropy and the mean photon number of the environment).  One may hope that the conditional entropy power inequality may also be used to address additivity problems in the context of entanglement-assisted communication. A natural candidate problem here is the rate region of the quantum
multiple access channel (MAC) characterized by Hsieh, Devetak and Winter~\cite{Hsetal}, where a single-letter formula is not known.  The use of conditional entropy power inequalities is especially suggestive in the case of the additive bosonic MAC (see e.g.,~\cite{YenShap05}), where Alice and Bob are connected to a receiver Charlie via two arms of a beamsplitter.  As shown by
 Cezkaj et al.~\cite{czetal10} (see also~\cite{czetalsec}), this scenario
exhibits uniquely quantum activation effects:
 providing entanglement-assistance to Bob can boost Alice's maximal rate of communication. This is in contrast to analogous classical settings, where providing additional resources to Bob cannot change her maximal rate (the latter  is determined by her signal power). While the conditional entropy power inequality can be used to bound the strength of this activation effect, a direct application does unfortunately not appear to yield fundamental new insights into the  additivity problem for the bosonic MAC.

\subsubsection*{Outline}
In Section~\ref{sec:basicdefinitions}, we recall basic definitions. In Section~\ref{sec:symplecticeigenvalue}, we perturbatively compute the first-order corrections to the symplectic eigenvalues of symmetric matrices relevant in our context. In Section~\ref{sec:condentropydiff}, we apply these perturbative results to obtain the asymptotic scaling of the conditional entropy, as well as its infinitesimal rate of increase when part of a Gaussian state undergoes diffusion.  In Section~\ref{sec:conditionalfisher}, we connect this to Fisher information by establishing a de Bruijin-type identity for conditional entropies.
 In Section~\ref{sec:conditionalentropypower}, we complete the proof of the entropy power inequality for conditional entropies.

\section{Basic definitions\label{sec:basicdefinitions}}
Consider $N$~bosonic modes described by mode operators $\vec{R}=(Q_1,P_1,\ldots,Q_N,P_N)$ satisfying the canonical commutation relations
\begin{align*}
[R_k,R_\ell]=iJ_{k,\ell}\qquad\textrm{ where }\qquad J_N=\left(\begin{matrix}
0 & 1\\ 
-1 & 0
\end{matrix}\right)^{\oplus N}\ .
\end{align*}
A  Gaussian quantum state~$\rho=\rho_{M,\vec{d}}$ on this system is completely specified by its first and second moments, i.e., the 
 displacement vector $\vec{d}=(d_1,\ldots,d_{2N})\in\mathbb{R}^{2N}$
and its (symmetric) covariance matrix~$M$, defined by
\begin{align*}
d_k=\tr(\rho R_k)\qquad\textrm{ and }\qquad  M_{j,k}=\tr(\rho\{R_j-d_j,R_k-d_k\})\ .
\end{align*}
Here $\{A,B\}=AB+BA$. We call $\rho_{M,0}$ a centered Gaussian state and often write
$\rho_M$ for it. 

A Gaussian operation 
maps Gaussian states to Gaussian states and is determined by its action on $(M,d)$. 
An example is a displacement (or Weyl) operator $W(\vec{\xi})$, $\vec{\xi}\in\mathbb{R}^{2N}$: this is a unitary operation satisfying
\begin{align}
W(\vec{\xi})\rho_{M,\vec{d}}W(\vec{\xi})^\dagger &=\rho_{M,\vec{d}+\vec{\xi}}\ ,\label{eq:displacementaction}
\end{align}
for all Gaussian states $\rho_{M,d}$. It will sometimes be convenient to write the conjugation map as $\cW_{\vec{\xi}}(\rho):=W(\vec{\xi})\rho W(\vec{\xi})^\dagger$ .
Statement~\eqref{eq:displacementaction} is equivalent to the  Heisenberg action on the mode operators 
\begin{align}
W(\xi)^\dagger R_k W(\xi)=R_k+\xi_k\qquad\textrm{ for all }k=1,\ldots,2N\ .\label{eq:heisenbergactiondisplacement}
\end{align}
A matrix $S$ satisyfing $SJ_NS^T=J_N$ is called symplectic. A symplectic matrix~$S$ uniquely defines a Gaussian unitary $U_S$
by the action on mode operators
\begin{align}
U_S^\dagger R_j U_S&=\sum_{k=1}^{2N} S_{j,k} R_k=:R_j'\qquad\textrm{ for all }j=1,\ldots,2N\ .\label{eq:transformationmode}
\end{align}
Because $S$ is symplectic, the transformed operators $\vec{R}'=(R_1',\ldots,R_{2N}')=(Q_1',P_1',\ldots,Q_N',P_N')$ again satisfy canonical commutation relations. It is convenient to define the (transformed) creation and annihilation operators
\begin{align*}
a_k^\dagger =\frac{1}{\sqrt{2}}(Q'_k+iP'_k)\qquad\textrm{ and }\qquad
a_k=\frac{1}{\sqrt{2}}(Q_k'-iP_k')\ .
\end{align*}
The associated number operators 
\begin{align}
\hat{n}_k=a_k^\dagger a_k=\frac{1}{2}((Q'_k)^2+(P'_k)^2)-\frac{1}{2}I,\qquad k=1,\ldots,N\ \label{eq:numberoperatorsev}
\end{align} are mutually commuting, and there is a simultaneous orthonormal eigenbasis of the form $\ket{n}=\ket{n_1,\ldots,n_k}$ with $\hat{n}_k\ket{n}=n_k\ket{n}$. 

 Eq.~\eqref{eq:transformationmode} translates
into the action
\begin{align}
U_S\rho_{M,d}U_S^\dagger &=\rho_{SMS^T,Sd}\ ,\label{eq:USactioncovariance}
\end{align}
on the covariance matrix and the displacement vector. In particular, displacement operators
and unitaries of the form $U_S$ can be used
to diagonalize any Gaussian state (with~\eqref{eq:displacementaction} and~\eqref{eq:USactioncovariance}). More precisely, Williamson's theorem~\cite{Williamson36}  states that a symmetric  positive definite matrix~$M$ can be diagonalized by a symplectic matrix~$S$, i.e.,
\begin{align}
SMS^T=\diag(\lambda_1,\lambda_1,\lambda_2,\lambda_2,\ldots,\lambda_N,\lambda_N)=:D_N(\vec{\lambda})\  .\label{eq:symplecticdiagonalizationgamma}
\end{align}
We call $\vec{\lambda}=(\lambda_1,\ldots,\lambda_N)$ the symplectic eigenvalues of~$M$. Observe that this list may include multiplicities (i.e., $\lambda_i=\lambda_j$ for $i\neq j$). We will occasionally denote the set of distinct symplectic eigenvalues
as~$\sspec(M)$. The quantity~$\Delta(M)=\min_{\lambda,\tilde{\lambda}\in \sspec(M), \lambda\neq \tilde{\lambda}}|\lambda-\tilde{\lambda}|$ will be referred to as the symplectic gap of~$M$. 

With~\eqref{eq:USactioncovariance}, identity~\eqref{eq:symplecticdiagonalizationgamma} implies that a centered Gaussian state $\rho_M$ with covariance matrix $M$ can be brought into product form as
\begin{align}
U_S\rho_M U_S^\dagger =\bigotimes_{j=1}^n \frac{e^{-\beta_j \hat{n}_j}}{\tr e^{-\beta_j\hat{n}_j}}\ ,\label{eq:usrhodiagonalization}
\end{align}
where the inverse temperatures~$\beta_j=\beta(\lambda_j)$ are given in terms of the symplectic eigenvalues by
\begin{align}
\beta(\lambda)=\log \frac{\lambda+1}{\lambda-1}\ .\label{eq:inversetemperaturedef}
\end{align}
Note  that $\beta(\lambda)$ is monotonically decreasing
with increasing $\lambda$, $\beta(\lambda)<2$ for $\lambda>2$, and
$\lim_{\lambda\rightarrow\infty}\beta(\lambda)=0$.  If $\lambda_j=1$, then the factor~$\frac{e^{-\beta_j \hat{n}_j}}{\tr e^{-\beta_j\hat{n}_j}}$ in~\eqref{eq:usrhodiagonalization} needs to be replaced by the pure `vacuum' state
$\proj{0}_{Q_j'P_j'}$ associated the the mode operators~$Q'_j$, $P'_j$.  Eq.~\eqref{eq:usrhodiagonalization} shows that the number states~$\{\ket{n}\}_{n\in\mathbb{N}_0^N}$ corresponding to the transformed mode operators~$\vec{R}'$ are an eigenbasis of~$U_S\rho_MU_S^\dagger$.

The entropy $S(\rho)=-\tr(\rho\log\rho)$ of an  $N$-mode Gaussian state~$\rho=\rho_{M,d}$ 
only depends on the symplectic eigenvalues $\vec{\lambda}=(\lambda_1,\ldots,\lambda_N)$ of the covariance matrix. To express it, it is useful to define the mean photon number~$\photonnumber(\lambda_k)$ in the eigenmode~$k$ by the function
\begin{align}
\photonnumber(\lambda)=(\lambda-1)/2\label{eq:meannumbersymplectic}
\end{align}
of a symplectic eigenvalue~$\lambda$. Then the entropy is given by
\begin{align}
S(\rho)&=\sum_{j=1}^N g(\photonnumber(\lambda_j))\qquad\textrm{ with }\qquad g(\photonnumber):=(\photonnumber+1)\log (\photonnumber+1)-\photonnumber\log \photonnumber\ .\label{eq:entropydefinition}
\end{align}
Let us give a  simple bound on the dependence of the entropy on the symplectic eigenvalues.
\begin{lemma}\label{lem:entropycontinuity}
Suppose $\rho$ and $\sigma$ are $N$-mode Gaussian states with  (arbitrarily ordered)
symplectic eigenvalues $\vec{\nu}=(\nu_1,\ldots,\nu_N)$ 
and $\vec{\lambda}=(\lambda_1,\ldots,\lambda_N)$, respectively, 
and assume that $\lambda_j\neq 1$ for all $j=1,\ldots,N$. 
Let $\|\vec{\nu}-\vec{\lambda}\|_\infty=\max_{1\leq j\leq N} |\nu_j-\lambda_j|$
and $\lambda_*=\min_{1\leq j\leq N}\lambda_j$.
Then
\begin{align*}
|S(\rho)-S(\sigma)|\leq
\frac{N}{2}\left(\|\vec{\nu}-\vec{\lambda}\|_\infty\beta(\lambda_*)+
\frac{\|\vec{\nu}-\vec{\lambda}\|_\infty^2}{\lambda_*^2-1}\right)\ .
\end{align*}
\end{lemma}
\noindent The requirement that~$\sigma$ has no symplectic eigenvalue equal to~$1$ (or equivalently $\lambda_*>1$) implies that none of the eigenmodes factors out in a pure product state, and we can usually assume this without loss of generality in our considerations below.

\begin{proof}
We have $g'(\photonnumber)=\log ((\photonnumber+1)/\photonnumber)$, and combining
this with~\eqref{eq:meannumbersymplectic}, we obtain the inverse temperature  according to
\begin{align}
g'(\photonnumber(\lambda))=\beta(\lambda)\ .\label{eq:gprimeNgamma}
\end{align}
Furthermore, we
have
\begin{align}
\beta'(\lambda)=-\frac{2}{\lambda^2-1}\ .\label{eq:betaderivative}
\end{align}
Let us define the function $G(\lambda):=g(\photonnumber(\lambda))$. Because $\photonnumber'(\lambda)=\lambda/2$, 
Eqs.~\eqref{eq:gprimeNgamma} and~\eqref{eq:betaderivative} imply
\begin{align*}
G'(\lambda)=\beta(\lambda)/2\qquad\textrm{ and }\qquad G''(\lambda) =-\frac{1}{\lambda^2-1}\ .
\end{align*}
Hence the Taylor series expansion gives
\begin{align*}
|G(\lambda+\epsilon)-G(\lambda)|\leq \epsilon \beta(\lambda)/2 + \frac{\epsilon^2}{2|\lambda^2-1|}\ .
\end{align*}
Since $\beta$ is monotonically decreasing, we get
\begin{align*}
\max_j |G(\nu_j)-G(\lambda_j)|\leq \frac{1}{2}\max_k |\nu_k-\lambda_k| 
\beta(\min_j\lambda_j) + \frac{\max_k |\nu_k-\lambda_k|^2}{2\min_j|\lambda_j^2-1|}\ 
\end{align*}
The claim follows from this because
$S(\rho)=\sum_{j=1}^N G(\nu_j)$ and $S(\sigma)=\sum_{j=1}^N G(\lambda_j)$\ .
\end{proof}

\section{Perturbation theory for symplectic eigenvalues\label{sec:symplecticeigenvalue}}
We begin with a straightforward application of degenerate perturbation theory
which is illustrative of the method. We will need slightly more involved statements for bipartite systems below (Lemma~\ref{lem:firstorderdegenerate}). Note that a similar
perturbative analysis was used in a different context 
 in~\cite[Appendix B]{SerafiniEisertWolf05}.

\begin{lemma}[Perturbation of symplectic spectrum to $0$th order]\label{lem:zerothorder}
Let $M$ be a covariance matrix of $N$ modes and consider the 
symplectic eigenvalues $\vec{\lambda}(\epsilon)=(\lambda_1(\epsilon),\ldots,\lambda_N(\epsilon))$ of 
\begin{align*}
M^\infty(\epsilon)&=I+\epsilon M\ ,
\end{align*}
where  we assume that $\|M\|=O(1)$ and $\epsilon\ll 1$.  Then
\begin{align*}
\lambda_j(\epsilon)=1+O(\epsilon)\qquad\textrm{ for all }j=1,\ldots, N\ .
\end{align*}
\end{lemma}
\begin{proof}
The symplectic eigenvalues $\vec{\lambda}(\epsilon)$ can be obtained from the spectrum~$\spec(iJ_NM^\infty(\epsilon))$
of  $iJ_NM^\infty(\epsilon)$  since this  operator has eigenvalues $(\lambda_1(\epsilon),-\lambda_1(\epsilon),\ldots,\lambda_N(\epsilon),-\lambda_N(\epsilon))$. Let us write
\begin{align*}
iJ_NM^\infty(\epsilon)&=H+V\qquad\textrm{ where }\qquad H=iJ_N\qquad\textrm{ and }\qquad V=i\epsilon J_NM\ .
\end{align*} 
Observe that the eigenvalues of $H$ are $\{1,-1\}$ and $\|V\|=O(\epsilon)$. This means that we can apply standard degenerate perturbation theory: the spectrum $\spec(H+V)$ is given by
\begin{align}
\{\mu+\spec(V|_{\cH_\mu})\ \big|\ \mu\in\spec(H)\}+O(\epsilon^2)\label{eq:firstorderdegenerateperturb}
\end{align}
to first order in $\epsilon$, where $V|_{\cH_\mu}$ denotes the restriction of $V$ to the degenerate eigenspace~$\cH_\mu$ of~$H$ to eigenvalue~$\mu$. An immediate consequence is that
(by assumption on $\|M\|$), the symplectic eigenvalues  of $M^\infty(\epsilon)$ are of the form
$\lambda_j(\epsilon)=1+O(\epsilon)$ for all $j=1,\ldots,N$. 
\end{proof}
To obtain the correct constant (i.e., the first-order correction in~\eqref{eq:firstorderdegenerateperturb}), it is necessary to compute the restriction $V|_{\cH_\mu}$ to the eigenspace with eigenvalue~$\mu$. For this purpose, it is convenient to use a basis of $\cH_\mu$.  For example, for the case of Lemma~\ref{lem:zerothorder}, we can define the vectors
\begin{align*}
\ket{v^+}=\frac{1}{\sqrt{2}}(i,1)=\frac{1}{\sqrt{2}}(i\ket{1}+\ket{2})\qquad\textrm{ and }\qquad 
\ket{v^-}=\frac{1}{\sqrt{2}}(1,i)=\frac{1}{\sqrt{2}}(\ket{1}+i\ket{2})\ 
\end{align*}
in $\mathbb{R}^2$ (Here and below, we will often use $\ket{j}$ to denote standard orthonormal basis vectors in~$\mathbb{R}^{2N}$.) Importantly, these vectors satisfy
\begin{align}
iJ_1 \ket{v^\pm}=\pm \ket{v^\pm}\ ,\label{eq:joneeigenspace}
\end{align}
and 
\begin{align}
\bra{v^\pm} iJ_1 D_1(\alpha)\ket{v^\pm}=\pm\alpha\qquad \bra{v^\mp} iJ_1 D_1(\alpha)\ket{v^\pm}=0\qquad\textrm{ for all }\alpha\in\mathbb{R}\  .\label{eq:donesandwich}
\end{align}
Writing $(0,0)^{\oplus k}=(0,0)\oplus\cdots\oplus (0,0)$ ($k$ summands), we can define
\begin{align*}
\ket{v_{j,N}^{\pm}}=(0,0)^{(j-1)}\oplus \ket{v^{\pm}}\oplus(0,0)^{\oplus(N-j-1)}\ .
\end{align*}
where $\ket{v^{\pm}}$ is at mode~$j$ (of $N$~modes). 

Because of~\eqref{eq:joneeigenspace}, the eigenspace $\cH_1$  of $H$ to eigenvalue $1$ has orthonormal basis $\{\ket{v_{j,N}^+}\}_{j=1,\ldots,N}$. Similarly, the eigenspace $\cH_{-1}$ has orthonormal basis $\{\ket{v_{j,N}^-}\}_{j=1,\ldots,N}$. The restriction of $V$ to $\cH_1$ is given by the matrix $(\bra{v_j^+}V\ket{v_k^+})_{j,k=1,\ldots,N}$, and diagonalizing this matrix gives the first-order corrections to the eigenvalue~$1$ (the degeneracy is lifted). We will omit a more detailed discussion here as we will need a more general version (including a reference system~$B$). However, the proof of the following Lemma proceeds in  this fashion; the key here is to apply perturbation theory to a suitably transformed matrix.
\begin{lemma}[Perturbation of symplectic spectrum to $1$st order]\label{lem:firstorderdegenerate}
Let
\begin{align*}
M_{AB}&=\left(
\begin{matrix}
M_A & L_{AB}\\
L_{AB}^T & M_B
\end{matrix}
\right)
\end{align*}
be the covariance matrix of a state on~$m+n$~modes with respect to the ordering
\begin{align*}
\vec{R}=(Q_1^A,P_1^A,\ldots,Q_m^A,P_m^A,Q_1^B,P_1^B,\ldots,Q_n^B,P_n^B) 
\end{align*}
of modes. 
\begin{enumerate}[(i)]
\item\label{it:infinitetime}
Assume that $\|M_{AB}\|=O(1)$ and $\epsilon\ll 1$. Then
the symplectic eigenvalues~$\vec{\lambda}(\epsilon)=(\lambda_1(\epsilon),\ldots,\lambda_{m+n}(\epsilon))$ 
of the covariance matrix
\begin{align}
M^{\infty}_{AB}(\epsilon)&=
\left(
\begin{matrix}
I_A & 0\\
0& 0
\end{matrix}
\right)+\epsilon M_{AB}=:I_A\oplus 0_B+\epsilon M_{AB}
\end{align}
are of the form
\begin{align*}
\lambda_{j}(\epsilon)=\eta_j(\epsilon)+O(\epsilon^2)\ \textrm{ for }j=1,\ldots,m\quad\textrm{ and }\quad 
\lambda_{m+j}(\epsilon)=\epsilon\nu_j+O(\epsilon^2)\ \textrm{ for }j=1,\ldots,n\ 
\end{align*}
where $\vec{\eta}(\epsilon)=(\eta_1(\epsilon),\ldots,\eta_m(\epsilon))$ 
are the eigenvalues of $I_A+\epsilon M_A$ and 
$\vec{\nu}=(\nu_1,\ldots,\nu_n)$ are the eigenvalues of $M_B$.
\item\label{it:infinitesimaltime}
Let  $\vec{\lambda}=(\lambda_1,\ldots,\lambda_{m+n})$ be the symplectic eigenvalues of $M_{AB}$  and let $S_{AB}$ be the symplectic matrix diagonalizing $M_{AB}$, i.e.,
\begin{align*}
S_{AB}M_{AB}S_{AB}^T&=D_{m+n}(\vec{\lambda})\  .
\end{align*}
 Suppose  $\epsilon\ll \Delta(M_{AB})$. Let
$\vec{\lambda}(\epsilon)=(\lambda_1(\epsilon),\ldots,\lambda_{m+n}(\epsilon))$  be the symplectic eigenvalues of
\begin{align}
M^{0}_{AB}(\epsilon)&=M_{AB}+
\epsilon I_A\oplus 0_B\  .
\end{align}
Then degeneracies are split according to
\begin{align}
|\{\ell\ |\ \lambda_\ell=\lambda\}|=|\{\ell\ |\ 
\lambda_\ell(\epsilon)\in [\lambda-\Delta/2,\lambda+\Delta/2]\}\label{eq:numberofeigenvaluesdeglambda}
\end{align}
where $\ell$ ranges over $\ell \in \{1,\ldots,m+n\}$. Furthermore
\begin{align*}
\sum_{\ell:\lambda_\ell(\epsilon)\in [\lambda-\Delta/2,\lambda+\Delta/2]}\lambda_\ell(\epsilon)
&=\sum_{\ell:\lambda_\ell=\lambda} \left(\lambda+\frac{\epsilon}{2}
\tr  [S_{AB}(I_A\oplus 0_B)S_{AB}^T]^{(\ell)}\right)+O(\epsilon^2)\
\end{align*}
for each $\lambda\in\sspec(M_{AB})$. Here $[Z]^{(\ell)}$ denotes the submatrix corresponding to the $\ell$-th mode, i.e., the $2\times 2$-matrix with entries $([Z]^{(\ell)})_{i,j}=Z_{2\ell-1+i,2\ell-1+j}$, $i,j\in\{0,1\}$. 
\end{enumerate}
\end{lemma}

\subsection*{Proof of Lemma~\ref{lem:firstorderdegenerate}~\eqref{it:infinitetime}}
Our goal is to compute the symplectic eigenvalues~$\lambda_1(\epsilon),\ldots,\lambda_{m+n}(\epsilon)$ 
of the covariance matrix
\begin{align*}
M^\infty_{AB}(\epsilon)&=
I_A\oplus 0_B+\epsilon M_{AB}\  ,
\end{align*}
to first order in~$\epsilon$. 
Let $\vec{\nu}=(\nu_1,\ldots,\nu_m)$ 
be the symplectic eigenvalues of $M_B$, and 
let $S_B$ be the symplectic matrix diagonalizing $M_B$, i.e.,
\begin{align*}
S_B M_B S_B^T&=D_n(\vec{\nu})\ ,
\end{align*}
where $D_n(\vec{\nu})=\diag(\nu_1,\nu_1)\oplus\cdots\oplus\diag(\nu_n,\nu_n)$.
Similarly, let $\vec{\alpha}(\epsilon)=(\alpha_1(\epsilon),\ldots,\alpha_m(\epsilon))$ be
such that the symplectic eigenvalues of $I_A+\epsilon M_A$ are
equal to $(1+\alpha_1(\epsilon),\ldots,1+\alpha_m(\epsilon))$. Observe that
\begin{align}
|\alpha_j(\epsilon)|=O(\epsilon)\qquad\textrm{ for }j=1,\ldots,m\ .\label{eq:mujepsilonperturb}
\end{align}
Let~$S_A(\epsilon)$ be the symplectic matrix diagonalizing $I_A+\epsilon M_A$, i.e., 
\begin{align*}
S_A(\epsilon)(I_A+\epsilon M_A)S_A(\epsilon)^T&=I_A+D_m(\vec{\alpha}(\epsilon))\ .
\end{align*}
Finally, let $S_{AB}(\epsilon)=S_A(\epsilon)\oplus S_B$.  The symplectic spectrum of~$M^\infty_{AB}(\epsilon)$ is identical to that of
\begin{align*}
\hat{M}^\infty_{AB}(\epsilon):=S_{AB}(\epsilon)M^\infty_{AB}(\epsilon)S_{AB}(\epsilon)^T=
\left(
\begin{matrix}
I_A+D_m(\vec{\alpha}(\epsilon)) & \epsilon S_A(\epsilon) L_{AB}S_B^T\\
\epsilon S_BL^T_{AB}S_A(\epsilon)^T & \epsilon D_n(\vec{\nu})
\end{matrix}
\right)\ .
\end{align*}
In particular, the eigenvalues of the matrix~$iJ_{m+n}\hat{M}^\infty_{AB}(\epsilon)$
are given by~\begin{align}
\spec(iJ_{m+n}\hat{M}_{AB}^\infty(\epsilon))=(\lambda_1(\epsilon),-\lambda_1(\epsilon),\ldots,\lambda_{m+n}(\epsilon),-\lambda_{m+n}(\epsilon))\ .\label{eq:symplecticspectr}
\end{align}
Since $J_{m+n}=J_m\oplus J_n$,  we obtain
\begin{align*}
iJ\hat{M}^\infty_{AB}(\epsilon)=H+V\qquad\textrm{ where }H=\left(
\begin{matrix}
iJ_m & 0 \\
0 & 0
\end{matrix}
\right)\ \textrm{ and }\ 
V=
\left(
\begin{matrix}
iJ_m D_m(\vec{\alpha}(\epsilon)) & J_mX(\epsilon)\\
J_nX(\epsilon)^T & i\epsilon J_n D_n(\vec{\nu})
\end{matrix}
\right)\ ,
\end{align*}
where we introduced the abbreviation~$X(\epsilon)=i\epsilon S_A(\epsilon) L_{AB}S_B^T$.  By the assumption~$\|M_{AB}\|=O(1)$, we have $\|X(\epsilon)\|=O(\epsilon)$ and $\|D_n(\vec{\nu})\|=O(1)$. Similarly, we have $\|D_m(\vec{\alpha}(\epsilon))\|=O(\epsilon)$ because of Eq.~\eqref{eq:mujepsilonperturb}. We conclude that $\|V\|=O(\epsilon)$.  On the other hand, the eigenvalues of~$H$ are~$\{0,1,-1\}$, hence we can apply first-order (degenerate) perturbation theory to compute the spectrum of~$iJ\hat{M}^\infty_{AB}(\epsilon)$.

 We consider the eigenspaces separately.
\begin{itemize}
\item  $\mu=0$: Consider the eigenspace~$\cH_{0}$ of $H$ to eigenvalue~$\mu=0$.
It is easy to check that the vectors $\{v^+_{1,B},
v^-_{1,B},\ldots,
v^+_{n,B},
v^-_{n,B}\}$ 
defined by
\begin{align*}
\ket{v_{j,B}^{\pm}}:=0_A\oplus \ket{v_{j,n}^\pm }\ 
\end{align*}
are an orthonormal basis of~$\cH_0$. (Here we write
$0_A$ for the zero-vector $(0,0)^{\oplus m}$.)
Since
\begin{align*}
V\ket{v_{k,B}^{\tau}}&= (J_mX\ket{v_{k,n}^\tau})\oplus (i\epsilon J_nD_n(\nu)\ket{v_{k,n}^\tau})\ ,
\end{align*}
we get the matrix elements
\begin{align*}
\bra{v_{j,B}^{\sigma}}V\ket{v_{k,B}^{\tau}}&=\bra{v_{j,n}^\sigma}i\epsilon J_n D_n(\vec{\nu})\ket{v_{k,n}^\tau}=\delta_{j,k}\delta_{\sigma,\tau}(\sigma\cdot\epsilon \nu_j)\ 
 \end{align*}
according to~\eqref{eq:donesandwich}. In particular, the restriction~$V|_{\cH_0}$ is described by a diagonal matrix with diagonal elements of the form $\{\pm\epsilon \nu_j\}_{j=1}^n$. With~\eqref{eq:firstorderdegenerateperturb}, we have found~$m$ symplectic eigenvalues of the form
\begin{align}
\lambda_{m+j}(\epsilon)=\epsilon \nu_j+O(\epsilon^2)\qquad\textrm{ for }j=1,\ldots,n\ .\label{eq:lambdajepsnu}
\end{align}
(Only the non-negative entries on the diagonal are symplectic eigenvalues according to~\eqref{eq:symplecticspectr}.)
Note that these are simply the eigenvalues of~$M_B$, up to a factor of~$\epsilon$.
\item
$\mu=1$: An orthonormal basis of $\cH_1$ is given by the vectors~$\{\ket{v_{j,A}^+}\}_{j=1}^m$ defined by
\begin{align*}
\ket{v_{j,A}^{+}}:=\ket{v_{j,m}^\pm }\oplus 0_B\ 
\end{align*}
where   $0_B$ stands for $(0,0)^{\oplus n}$.
It is easy to check that
\begin{align*}
\bra{v_{j,A}^{+}} V\ket{v_{k,B}^{+}}=\delta_{j,k}\alpha_j(\epsilon)\ ,
\end{align*}
i.e., the restriction of $V$ to~$\cH_1$ is diagonal when expressed in this basis. In conclusion, we found~$m$ symplectic eigenvalues of the form 
\begin{align}
\lambda_{j}(\epsilon)=1+\alpha_j(\epsilon)+O(\epsilon^2)\qquad\textrm{ for }j=1,\ldots,m\ .\label{eq:lambdamjdef}
\end{align}
By definition of $\alpha_j(\epsilon)$, these are equal to  the symplectic eigenvalues~$\eta_j(\epsilon)=1+\alpha_j(\epsilon)$ of $I_A+\epsilon M_A$ in order~$O(\epsilon)$. 

\item
$\mu=-1$: Here we find an orthonormal basis with vectors~$\ket{v_{j,A}^{+}}:=\ket{v_{j,m}^+}\oplus 0_B$. The restriction of $V$ to~$\cH_{-1}$ is diagonal
with diagonal entries~$(-\alpha_1(\epsilon),\ldots,-\alpha_m(\epsilon))$. The corresponding eigenvalues are the negatives of~\eqref{eq:lambdamjdef}, consistent with~\eqref{eq:symplecticspectr}.
\end{itemize}

\subsection*{Proof of Lemma~\ref{lem:firstorderdegenerate}~\eqref{it:infinitesimaltime}}
Let us briefly recall the definitions involved in the statement. We consider the covariance matrix
\begin{align*}
M^{0}_{AB}(\epsilon)&=M_{AB}+
\epsilon\left(
\begin{matrix}
I_A & 0\\
0& 0
\end{matrix}
\right)=M_{AB}+\epsilon I_A\oplus 0_B\ ,
\end{align*}
where $M_{AB}$  has 
symplectic eigenvalues~$\vec{\lambda}=(\lambda_1,\ldots,\lambda_{m+n})$
and is diagonalized by $S_{AB}$, i.e., 
\begin{align*}
S_{AB}M_{AB}S_{AB}^T=D_{m+n}(\vec{\lambda})\ .
\end{align*}
We assume that $\epsilon\ll \Delta$, where $\Delta=\Delta(M_{AB})$ is the symplectic gap. 
Then $M^0_{AB}(\epsilon)$ has the same symplectic eigenvalues as
\begin{align*}
\hat{M}^0_{AB}(\epsilon)=D_{m+n}(\vec{\lambda})+\epsilon S_{AB}(I_A\oplus 0_B) S_{AB}^T\ .
\end{align*}
In particular, it suffices to find the eigenvalues of
\begin{align*}
iJ_{m+n}\hat{M}^0_{AB}(\epsilon)=iJ_{m+n}D_{m+n}(\vec{\lambda})+i\epsilon J_{m+n}S_{AB}(I_A\oplus 0_B)S_{AB}^T=:H+V\ .
\end{align*}
Because $H:=iJ_{m+n}D_{m+n}(\vec{\lambda})$ has gap $\Delta$, $\|iJ_{m+n}S_{AB}(I_A\oplus 0_B)S_{AB}^T\|=O(1)$ and our assumption $\epsilon\ll \Delta$, we can apply (degenerate) perturbation theory to compute the spectrum of~$iJ\hat{M}^0_{AB}(\epsilon)$.

The operator~$H$ has eigenvectors $(w_1^+,\ldots,w_{m+n}^+)=(v_{1,A}^+,\ldots,v_{m,A}^+,v_{1,B}^+,\ldots,v_{n,B}^+)$ with eigenvalues $\lambda_1,\ldots,\lambda_{m+n}$.
 In particular, for $\lambda\in\sspec(M_{AB})$, the eigenspace~$\cH_\lambda$ of~$H$ is 
\begin{align*}
\cH_\lambda &=\myspan \{\ket{w_\ell^+}\ |\ 1\leq \ell\leq m+n\textrm{ with } \lambda_\ell=\lambda\}\ .
\end{align*}Suppose that
 restriction~$V|_{\cH_{\lambda}}$ of $V:=i\epsilon J_{m+n}S_{AB}(I_A\oplus 0_B) S_{AB}^T$ to~$\cH_{\lambda}$  has eigenvalues~$\theta_1,\ldots,\theta_{\dim\cH_{\lambda}}$. According to
(degenerate) first-order perturbation theory, the matrix $iJ_{m+n}\hat{M}^0_{AB}(\epsilon)$ 
has~$\dim \cH_\lambda$ eigenvalues of the form
\begin{align}
\lambda+\theta_j+O(\epsilon^2)\qquad\textrm{ for }j=1,\ldots,\dim \cH_\lambda\ ,\label{eq:listgammathetaj}
\end{align}
where $\theta_j=O(\epsilon)$. Furthermore, the list~\eqref{eq:listgammathetaj}
includes all eigenvalues~$\lambda_\ell(\epsilon)$ of $iJ_{m+n}\hat{M}^0_{AB}(\epsilon)$ in the interval~$[\lambda-\Delta/2,\lambda+\Delta/2]$ and the number of such eigenvalues is equal to~\eqref{eq:numberofeigenvaluesdeglambda} by the assumption $\epsilon\ll \Delta$. We conclude that
\begin{align}
\sum_{\ell:\lambda_\ell(\epsilon)\in [\lambda-\Delta/2,\lambda+\Delta/2]}\lambda_\ell(\epsilon)
&=(\dim \cH_\lambda)\cdot\lambda+\sum_{j=1}^{\dim \cH_\lambda} \theta_j+O(\epsilon^2)\nonumber\\
&=(\dim \cH_\lambda)\cdot\lambda+\tr(V|_{\cH_\lambda})+O(\epsilon^2)\label{eq:dimHtrace}
\end{align}
Because the restriction~$V|_{\cH}$ can be expressed as
\begin{align*}
V|_{\cH_{\lambda}}&=\sum_{j,k:\lambda_j=\lambda_k=\lambda}
\bra{w^+_j}V\ket{w^+_k}\cdot \ket{w^+_j}\bra{w^+_k}\ ,
\end{align*}
we obtain
\begin{align}
\tr(V|_{\cH_\lambda})&=i\epsilon \sum_{\ell: \lambda_\ell=\lambda}\bra{w_{\ell}^+}J_{m+n}S_{AB}(I_A\oplus 0_B)S_{AB}^T\ket{w_{\ell}^+}\ .\label{eq:tracerestriction}
\end{align} 
Combining~\eqref{eq:tracerestriction} with 
$\dim\cH_\lambda=|\{\ell: \lambda_\ell=\lambda\}|$ and~Eq.~\eqref{eq:dimHtrace} gives
\begin{align}
\sum_{\ell:\lambda_\ell(\epsilon)\in [\lambda-\Delta/2,\lambda+\Delta/2]}\lambda_\ell(\epsilon)
=\sum_{\ell:\lambda_\ell=\lambda}
\left(
\lambda+i\epsilon \bra{w_{\ell}^+}J_{m+n}S_{AB}(I_A\oplus 0_B)S_{AB}^T\ket{w_{\ell}^+}
\right)+O(\epsilon^2)\ .\label{eq:gammadeltamn}
\end{align}
Since $\ket{w_\ell^+}=\frac{1}{\sqrt{2}}(i\ket{2\ell-1}+\ket{2\ell})$
and $J_{m+n}=\sum_{k=1}^{m+n} \ket{2k-1}\bra{2k}-\ket{2k}\bra{2k-1}$,  
we have $\bra{w_\ell^+}J_{m+n}=\frac{1}{\sqrt{2}}(i\bra{2\ell}+\bra{2\ell-1})$, and
this takes the form
\begin{align*}
\bra{w_{\ell}^+}J_{m+n}S_{AB}(I_A\oplus 0_B)S_{AB}^T\ket{w_{\ell}^+}
=\frac{1}{2} \left(iZ_{2\ell,2\ell-1}-iZ_{2\ell-1,2\ell}+Z_{2\ell,2\ell}+Z_{2\ell-1,2\ell-1}\right)\
\end{align*}
where $Z=S_{AB}(I_A\oplus 0_B)S_{AB}^T$. 
Since $Z^T=Z$ is symmetric, this is equal to
\begin{align}
\bra{w_{\ell}^+}J_{m+n}S_{AB}(I_A\oplus 0_B)S_{AB}^T\ket{w_{\ell}^+}=\frac{1}{2} \tr [S_{AB}(I_A\oplus 0_B)S_{AB}^T]^{(\ell)}\ ,\label{eq:finalexpression}
\end{align}
and the claim follows by inserting~\eqref{eq:finalexpression} into~\eqref{eq:gammadeltamn}. 
\section{Conditional entropy and diffusion\label{sec:condentropydiff}}

A central tool in the proof of the quantum entropy power inequality~\cite{KoeGrae} is the
one-parameter semigroup~$\{e^{t\cL}\}_{t\geq 0}$ of completely positive trace-preserving maps generated by the `diffusion' Liouvillean~$\cL$ (see~\cite{KoeGrae} for a definition of the latter). 
For all $t\geq 0$, the map $e^{t\cL}$ is Gaussian and can therefore be defined in terms of its action on the covariance matrix~$M$ and the displacement vector~$\vec{d}$: A Gaussian state
$\rho$ described by $(M,\vec{d})$ is transformed into a Gaussian state $\rho(t)=e^{t\cL}(\rho)$ with covariance matrix $(M(t)=M+tI,\vec{d})$, i.e., the
transformation  governing the evolution for time~$t$ is
\begin{align*}
\begin{matrix}
\rho & \mapsto &  e^{t\cL}(\rho)\\
(M,\vec{d})&\overset{e^{t\cL}}{\mapsto} &(M+tI,\vec{d})
\end{matrix}\qquad\textrm{ for all covariance matrices }M\textrm{ and displacement vectors }\vec{d}\ .
\end{align*}
In this section, we revisit and extend statements of~\cite{KoeGrae} about the behavior of the entropy~$S(e^{t\cL}(\rho))$ as a function of time~$t$. In particular,  we specialize to Gaussian initial states~$\rho$ and extend our considerations to conditional entropies. 

More precisely, we  consider the case where diffusion acts only on a subset of modes. Concretely, assume that our system is bipartite, with system~$A$ consisting of~$m$ modes, and system~$B$ consisting of $n$~modes. Diffusion for time~$t$ acting on the modes in~$A$ only is described by the superoperator~$e^{t\cL_A}\otimes I_B$ (where $I_B$ is the identity superoperator on $B$). In particular, this family of superoperators is 
specified by the transformation 
\begin{align}
\begin{matrix}
\rho_{AB} & \mapsto & (e^{t\cL_A}\otimes I_B)(\rho_{AB})=:\rho_{AB}(t)\\
(M_{AB},\vec{d}_{AB})&\overset{e^{t\cL_A}\otimes I_B}{\mapsto} &(M_{AB}+t(I_A\oplus 0_B),\vec{d}_{AB})
\end{matrix}\label{eq:diffusionprocessdef}
\end{align}
 for all covariance matrices $M_{AB}\in\Mat_{2(m+n)}(\mathbb{R})$ and 
displacement vectors~$\vec{d}_{AB}\in\mathbb{R}^{2(m+n)}$. We will examine
the conditional entropy~$S(A|B)_{\rho_{AB}(t)}$ for the evolved state
$\rho_{AB}(t)$, given some Gaussian initial state~$\rho_{AB}=\rho_{AB}(0)$. In Section~\ref{sec:infinitetimelimit}, we  show that $S(A|B)_{\rho_{AB}(t)}$ scales as a universal function for $t\rightarrow\infty$ (independent of the initial state~$\rho_{AB}$). In Section~\ref{sec:entropyincreasedef}, we derive an explicit expression for the infinitesimal rate of change of this quantity in terms of the covariance matrix of~$\rho_{AB}$.

\subsection{Scaling of the conditional entropy in the infinite-time limit\label{sec:infinitetimelimit}}
In the infinite-time-limit, the entropy 
of the time-evolved state~$e^{t\cL}(\rho)$ scales as a universal function of time~$t$ which is independent of the initial state~$\rho$. This statement was shown for general states in~\cite[Corollary 3.4]{KoeGrae}; here we give a simple argument for Gaussian states for completeness. We will also need this statement in the proof of Lemma~\ref{lem:conditionalentropyinfinitescaling} which deals with conditional entropies. 
\begin{lemma}[Scaling of entropy in the infinite-time limit under diffusion]
Let $\rho$ be an (arbitrary) Gaussian state of $N$ modes.\label{lem:unconditionalinfinitescaling}
Then
\begin{align*}
\lim_{t\rightarrow\infty}|S(e^{t\cL}(\rho))-N\cdot g((t-1)/2)|=0\ .
\end{align*}
\end{lemma}
\begin{proof}
Let $M$ be the covariance matrix of $\rho$. The covariance matrix of $\rho(t)=e^{t\cL}(\rho)$
has the form  $M+tI=tM^\infty(1/t)$. Therefore, the symplectic eigenvalues 
are  of the form $\nu_j=t(1+O(1/t))=t+O(1)$ according to Lemma~\ref{lem:zerothorder}. 
For $t\geq 1$, the matrix $tI$ is a valid covariance matrix with symplectic eigenvalues $\lambda_j=t$ for $j=1,\ldots,N$. Let~$\sigma(t)$ be the centered Gaussian state with covariance matrix~$tI$. 
Lemma~\ref{lem:entropycontinuity} applied to $\rho(t)$ and $\sigma(t)$ gives
\begin{align*}
|S(\rho(t))-S(\sigma(t))|\leq \frac{N}{2} O(1) \left(\beta(t)+\frac{1}{t^2-1}\right)\rightarrow 0\qquad\textrm{ for }t\rightarrow\infty\ .
\end{align*}
Since~$S(\sigma(t))=N\cdot g(\photonnumber(t))=N\cdot g((t-1)/2)$, the claim follows.
\end{proof}

A similar statement holds for conditional entropies:
\begin{lemma}[Scaling of conditional entropy in the infinite-time limit under diffusion]\label{lem:conditionalentropyinfinitescaling}
Let $\rho_{AB}$ be a Gaussian state of $(m+n)$~modes~$A$ and~$B$.
Define $\rho_{AB}(t)=(e^{t\cL_A}\otimes I_B)(\rho_{AB})$. 
Then 
\begin{align*}
\lim_{t\rightarrow \infty}|S(A|B)_{\rho_{AB}(t)}-m\cdot g((t-1)/2)|=0\ .
\end{align*}
\end{lemma}
\noindent Comparing Lemma~\ref{lem:conditionalentropyinfinitescaling} with Lemma~\ref{lem:unconditionalinfinitescaling} suggests that $(e^{t\cL_A}\otimes I_B)(\rho_{AB})$ approaches a product state for large times. Quantifying this convergence  for general (possibly non-Gaussian) states may  provide an avenue to proving conjecture~\ref{eq:qconditionalentropyinequality}.

\begin{proof}
Let $\vec{\nu}=(\nu_1,\ldots,\nu)$ be the symplectic eigenvalues of the covariance matrix~$M_B$. Note that we can assume without loss of generality that
\begin{align}
\nu_j\neq 1\qquad\textrm{ for all }j=1,\ldots,n\ .\label{eq:nondegenernuass}
\end{align}
Indeed, if we had $\nu_j=1$ for some~$j\in\{1,\ldots,n\}$, then the corresponding eigenmode in~$B$ is in a pure state and the state $\rho_{AB}$ factorizes. This is true also for the time-evolved state~$\rho_{AB}(t)$ (because $B$ is unaffected by the evolution), hence such eigenmodes do not contribute to the entropy~$S(A|B)_{\rho(t)}$ and can be traced out.

The covariance matrix of the state $\rho_{AB}(t)$ can be written in the form
\begin{align*}
M_{AB}+tI_A\oplus 0_B=t(I_A\oplus 0_B+\textfrac{1}{t}\cdot M_{AB})=:t M^\infty_{AB}(1/t)\ .
\end{align*}
Applying Lemma~\ref{lem:firstorderdegenerate}~\eqref{it:infinitetime} with $\epsilon=1/t$ (and multiplying the resulting eigenvalues by~$t$), we conclude that its symplectic eigenvalues are 
\begin{align*}
\lambda_{j}=\eta_j(t)+O(1/t)\ \textrm{ for } j=1,\ldots,m\ \textrm{ and }\ 
\lambda_{m+j}= \nu_j+O(1/t)\ \textrm{ for } j=1,\ldots,n\ , 
\end{align*}
where $\{\eta_j(t)\}$ are the symplectic eigenvalues of $tI_A+M_A$. This means that
up to order $O(1/t)$, the symplectic spectrum of $\rho_{AB}(t)$ 
is identical to the one associated with the product state $e^{t\cL}(\rho_A)\otimes \rho_B$. In summary, we have the list of symplectic eigenvalues
\begin{center}
\begin{tabular}{c|c|c}
$\rho_{AB}(t)$   & $e^{t\cL_A}(\rho_A)$ & $\rho_B$\\ \hline
$(\underbrace{\vec{\eta}(t)+O(1/t)}_{=:\tilde{\eta}}, \underbrace{\vec{\nu}(t)+O(1/t)}_{=:\tilde{\nu}})$ & $\vec{\eta}(t)$ & $\vec{\nu}(t)$
\end{tabular}
\end{center}

Let $\tilde{\rho}_A$ be a Gaussian state of $m$ modes with
symplectic spectrum $\tilde{\eta}$, and let $\tilde{\rho}_B$ be a Gaussian state of $n$ modes with symplectic spectrum~$\tilde{\nu}$. Then
\begin{align}
S(\rho_{AB}(t))=S(\tilde{\rho}_A)+S(\tilde{\rho}_B)\ .\label{eq:entropyrhotone}
\end{align}
We can apply Lemma~\ref{lem:entropycontinuity}
to the states $\tilde{\rho}_A$ and $e^{t\cL_A}(\rho_A)$.  Since the latter has 
symplectic eigenvalues $\vec{\eta}(t)=(\eta_1(t),\ldots,\eta_m(t))$, we get
\begin{align*}
|S(\tilde{\rho}_A)-S(e^{t\cL}(\rho_A))|\leq 
\frac{m}{2}\left(O(1/t)\beta(\eta_*(t))+O(1/t^2)\frac{1}{\eta_*(t)^2-1}\right)\rightarrow 0\qquad \textrm{ as }t\rightarrow \infty\ .
\end{align*}
Here we used the fact that $\eta_j(t)=t+O(1)$ for all $1\leq j\leq m$, hence $\eta_*(t)=\min_j \eta_j(t)\rightarrow\infty$ for $t\rightarrow\infty$. Similarly, applying Lemma~\ref{lem:entropycontinuity}
to the states $\tilde{\rho}_B$ and $\rho_B$  (and using assumption~\eqref{eq:nondegenernuass})
gives
\begin{align*}
|S(\tilde{\rho}_B)-S(\rho_B)|\leq \frac{n}{2}\left(
O(1/t)\cdot\beta(\nu_*)+O(1/t^2) \frac{1}{\min_j |\nu_j^2-1|}\right)\rightarrow 0\qquad\textrm{ for }t\rightarrow\infty\ .
\end{align*}
Inserting these upper bounds into~\eqref{eq:entropyrhotone} and using the triangle inequality then gives
\begin{align*}
|S((e^{t\cL_A}\otimes I_B)(\rho_{AB}))-S(e^{t\cL_A}(\rho_A)\otimes \rho_B)|\rightarrow 0\qquad\textrm{ for }t\rightarrow\infty\ .
\end{align*}
Because the states $\rho_{AB}(t)=(e^{t\cL_A}\otimes I_B)(\rho_{AB})$
and $\sigma_{AB}(t)=e^{t\cL_A}(\rho_A)\otimes \rho_B$ have the same reduced density operator $\rho_B(t)=\sigma_B(t)=\rho_B$, the previous statement implies that the difference between their conditional entropies also vanishes in the limit, that is, 
\begin{align}
|S(A|B)_{\rho_{AB}(t)}-S(A|B)_{\sigma(t)}|\rightarrow 0\qquad\textrm{ for }t\rightarrow\infty\ .\label{eq:sabrhot}
\end{align}
Since $\sigma_{AB}(t)$ is a product state, we have 
\begin{align}
S(A|B)_{\sigma(t)}=S(\sigma_A(t))\ .\label{eq:condentropyequal}
\end{align} The claim then follows from the triangle inequality, i.e.,
\begin{align*}
|S(A|B)_{\rho_{AB}(t)}-m\cdot g((t-1)/2)|\leq |S(A|B)_{\rho_{AB}(t)}-S(\sigma_A(t))|+|S(\sigma_A(t))-m\cdot g((t-1)/2)|
\end{align*}
because the first term on the rhs.~goes to~$0$ for $t\rightarrow\infty $ according to~\eqref{eq:sabrhot} and~\eqref{eq:condentropyequal}, whereas the second term goes to~$0$ according to Lemma~\ref{lem:unconditionalinfinitescaling}.
\end{proof}

\subsection{Rate of increase of the conditional entropy  \label{sec:entropyincreasedef}}
Next we compute the infinitesimal  rate  of increase of the conditional entropy under the process~\eqref{eq:diffusionprocessdef}.
\begin{lemma}[Rate of conditional entropy increase under diffusion]\label{lem:entropyincreasediff}
Consider bipartite system~$AB$ of  $m+n$ modes. 
Let $\rho_{AB}$ be a Gaussian quantum state
whose covariance matrix $M_{AB}$ has symplectic eigenvalues~$\vec{\lambda}=(\lambda_1,\ldots,\lambda_{m+n})$. Let $S_{AB}$ be a symplectic matrix such that  $S_{AB}M_{AB}S_{AB}^T=D_{m+n}(\vec{\lambda})$. Define~$\rho_{AB}(t)=(e^{t\cL_A}\otimes I_B)(\rho_{AB})$. 
Then 
\begin{align*}
\frac{d}{dt}\Big|_{t=0}S(A|B)_{\rho_{AB}(t)}
&=\frac{1}{4}\sum_{\ell=1}^{m+n}\beta(\lambda_\ell) \tr\left([S_{AB}(I_A\oplus 0_B) S_{AB}^T]^{(\ell)}\right)\ ,
\end{align*}
where $[Z]^{(\ell)}$ is the $2\times 2$~submatrix corresponding to the $\ell$-th mode, $([Z]^{\ell})_{i,j}=Z_{2\ell-1+i,2\ell-1+j}$ for $i,j\in\{0,1\}$. 
\end{lemma}
\begin{proof}
Because the reduced density operator $\tr_A\rho_{AB}(t)=\rho_B$ is independent of time, we have
\begin{align}
\frac{d}{dt}\Big|_{t=0}S(A|B)_{\rho_{AB}(t)}=\frac{d}{dt}\Big|_{t=0}S(AB)_{\rho_{AB}(t)}\ .\label{eq:evolvedsabt}
\end{align}
To evaluate the rate of change of $S(AB)_{\rho_{AB}(t)}$, let $(\lambda_1(t),\ldots,\lambda_{m+n}(t))$ be the symplectic eigenvalues of~$\rho_{AB}(t)$.  According to expression~\eqref{eq:entropydefinition} for the entropy, we have
\begin{align}
S\left(\rho_{AB}(t)\right)&=\sum_{j=1}^{m+n}g(\photonnumber(\lambda_j(t)))=\sum_{\lambda\in\sspec(M_{AB})} \sum_{\ell: \lambda_\ell(t)\in [\lambda-\frac{\Delta}{2},\lambda+\frac{\Delta}{2}]} g(\photonnumber(\lambda_\ell(t)))\ \label{eq:Srhoabexpl}
\end{align}
where we reexpressed the summation using the symplectic gap $\Delta=\Delta(M_{AB})$. 
With the expression~\eqref{eq:meannumbersymplectic}  for the mean photon number, Eq.~\eqref{eq:gprimeNgamma} and  the chain rule for differentiation, we have
\begin{align*}
g(\photonnumber(\lambda(t))=g(\photonnumber(\lambda(0)))+\frac{t\beta(\lambda(0))}{2}\cdot\lambda'(0)+O(t^2)\ 
\end{align*}
Observe that~$\rho_{AB}(t)$ has covariance matrix~$M^0_{AB}(t)$ and we can restrict our attention to times~$t\ll\Delta(M_{AB})$ without loss of generality. Hence we can apply Lemma~\ref{lem:firstorderdegenerate}~\eqref{it:infinitesimaltime}. We obtain 
\begin{align*}
\sum_{\ell: \lambda_\ell(t)\in [\lambda-\frac{\Delta}{2},\lambda+\frac{\Delta}{2}]}\!\!\!\!\!\!\!\!\! g(\photonnumber(\lambda_\ell(t)))&=g(\photonnumber(\lambda))\cdot |\{\ell\ |\lambda_\ell=\lambda\}|
+\frac{t\beta(\lambda)}{4}\sum_{\ell:\lambda_\ell=\lambda}
\tr [S_{AB}(I_A\oplus 0_B)S_{AB}^T]^{(\ell)}+O(t^2)\ .
\end{align*}
for any $\lambda\in\sspec(M_{AB})$. Taking the derivative at~$t=0$ and inserting into~\eqref{eq:Srhoabexpl}
therefore gives
\begin{align*}
\frac{d}{dt}\Big|_{t=0}S\left(\rho_{AB}(t)\right) &
=\frac{1}{4}\sum_{\lambda\in\sspec(M_{AB})}\!\!\!\!\!\!\!\!\! \beta(\lambda) \sum_{\ell:\lambda_\ell=\lambda}
\tr [S_{AB}(I_A\oplus 0_B)S_{AB}^T]^{(\ell)}\ ,
\end{align*}
which is the claim because of~\eqref{eq:evolvedsabt}.\end{proof}

\section{Diffusion, translations and Fisher information\label{sec:conditionalfisher}}
A key element in the proof of the classical entropy power inequality is de Bruijin's identity; it relates the infinitesimal rate of entropy increase to the Fisher information of a family of translated distributions. In~\cite{KoeGrae}, a quantum version of this statement in terms of the diffusion semigroup and phase space translations was given.  Here we derive a generalization of this statement for conditional entropies (but specialized to Gaussian states). Our proof   proceeds by direct calculation and does not involve any technical subtleties associated with formal computations involving infinite-dimensional systems. 

We begin by recalling the relevant definitions. Consider a one-parameter family $\{\rho^{(\theta)}\}_{\theta\in\mathbb{R}}$ of states depending smoothly on the parameter~$\theta$. 
The divergence-based Fisher information of this family (at $\theta_0\in\mathbb{R}$) is  defined as the quantity
\begin{align*}
J(\rho^{(\theta)};\theta)|_{\theta=\theta_0}=\frac{d^2}{d\theta^2} S(\rho^{(0)}\|\rho^{(\theta)})|_{\theta=\theta_0}\ ,
\end{align*}
where $S(\rho\|\sigma)=\tr(\rho\log\rho-\rho\log\sigma)$ is the relative entropy or divergence.  Two straightforward but important consequences of this definition 
are the reparametrization identities
\begin{align}
J(\rho^{(c\theta)};\theta)|_{\theta=0}=c^2J(\rho^{(\theta)};\theta)|_{\theta=0}\qquad\textrm{ and }\qquad J(\rho^{(\theta+c)};\theta)|_{\theta=\theta_0}=
J(\rho^{(\theta+c)};\theta)|_{\theta=\theta_0+c}\label{eq:reparametrizationidentity}
\end{align}
for $c\in\mathbb{R}$ and its additivity
\begin{align}
J(\rho_A^{(\theta)}\otimes\rho_B^{(\theta)};\theta)|_{\theta=\theta_0}=
J(\rho_A^{(\theta)};\theta)|_{\theta=\theta_0}+
J(\rho_B^{(\theta)};\theta)|_{\theta=\theta_0}\ .\label{eq:fisherinformationadditivity}
\end{align}
The latter follows from the additivity of the relative entropy under tensor products. Furthermore, because of the monotonicity~$S(\cE(\rho)\|\cE(\sigma))\leq S(\rho\|\sigma)$ of the relative entropy under CPTP maps~$\cE$ the Fisher information also satisfies monotonicity (see~\cite{KoeGrae}), i.e.,  
\begin{align}
J(\cE(\rho^{(\theta)});\theta)|_{\theta=0}\leq 
J(\rho^{(\theta)};\theta)|_{\theta=0}\ .\label{eq:fisherinformationmotonicity}
\end{align}

 de Bruijin's identity involves the family~$\{\rho^{(\theta,R_k)}=\cW_{\theta\ket{k}}(\rho)\}_{\theta\in\mathbb{R}}$ of states obtained by translating an $N$-mode state~$\rho$ in the direction~$R_k$ in phase space, where $k\in\{1,\ldots,2N\}$. Let us write
\begin{align}
J(\rho):=\sum_{k=1}^{2N} J(\rho^{(\theta,R_{k})};\theta)|_{\theta=0}\ \label{eq:Jrhouncond}
\end{align}
for the sum  of the corresponding Fisher informations. de Bruijin's identity (shown in~\cite{KoeGrae}) relates this to the rate of entropy increase under diffusion, i.e.,
\begin{align}
\frac{d}{dt}\Big|_{t=0}S(e^{t\cL}(\rho))=\frac{1}{4}J(\rho)\ .\label{eq:debruijinorig}
\end{align}
In this section, we derive a version of~\eqref{eq:debruijinorig} for Gaussian states which involves an auxiliary system: it quantifies the rate of increase in the conditional entropy~$S(A|B)$ when~$A$ undergoes diffusion. In contrast to~\cite{KoeGrae}, the proof given here proceeds by direct computation. In Section~\ref{sec:fisherinfotranslated}, we compute the Fisher information of a family of states obtained by translating a Gaussian in phase space.  In Section~\ref{sec:debruijin}, we combine this with Lemma~\ref{lem:entropyincreasediff} to prove the de Bruijin identity. 

\subsection{Conditional Fisher information of translated  Gaussian states\label{sec:fisherinfotranslated} }

Let $\rho_{M,\vec{d}}$ denote an $N$-mode  Gaussian state with covariance matrix $M$ and displacement~$\vec{d}\in\mathbb{R}^{2N}$. Suppose $S$ is a symplectic matrix such that $SMS^T=D(\vec{\gamma})$ is diagonal. Let $\vec{\theta}\in \mathbb{R}^{2N}$ be arbitrary.
We will need the following formula for the relative entropy of $\rho_{M,\vec{d}}$ and a displaced state $\rho_{M,\vec{d}+\vec{\theta}}$: we have
\begin{align}
S(\rho_{M,\vec{d}}\|\rho_{M,\vec{d}+\vec{\theta}})=F(\vec{\gamma})+\frac{1}{2}\sum_{j=1}^N \beta(\gamma_j)\left( (S\theta)^2_{2j-1}+(S\theta)^2_{2j}\right)\ .\label{eq:relativeentropydisplaced}
\end{align}
where $F$ is a function of the symplectic eigenvalues only. 
\begin{proof}
By the  invariance $S(U\rho U^\dagger\| U\sigma U^\dagger)=S(\rho\|\sigma)$ of the relative entropy under unitaries, and applying displacement operators as well as the unitary $U_S$, we have
\begin{align*}
S(\rho_{M,\vec{d}}\|\rho_{M,\vec{d}+\vec{\theta}})=S(\rho_{M,0}\|\rho_{M,\vec{\theta}})=S(\rho_{M,-\vec{\theta}}\|\rho_{M,0})=S(\rho_{D(\gamma),-S\vec{\theta}}\|\rho_{D(\gamma),0})\ .
\end{align*}
It hence suffices to analyze  $S(\rho_{D(\gamma),\vec{\eta}}\|\rho_{D(\gamma),0})$, where $\vec{\eta}=-S\vec{\theta}\in\mathbb{R}^{2N}$.
Because
\begin{align*}
\rho_{D(\gamma),\vec{\eta}}=\rho_{D_1(\gamma_1),(\eta_1,\eta_2)}\otimes\cdots\otimes \rho_{D_1(\gamma_N),(\eta_{2N-1},\eta_{2N})} 
\end{align*}
is a product state, we obtain (using the additivity of the relative entropy for product states) 
\begin{align*}
S(\rho_{M,d}\|\rho_{M,\vec{d}+\vec{\theta}})=S(\rho_{D(\gamma),\vec{\eta}}\|\rho_{D(\gamma),0})=\sum_{j=1}^N S(\rho_{D_1(\gamma_j), (\eta_{2j-1},\eta_{2j})}\|\rho_{D_1(\gamma_j),0})\ .
\end{align*}
The claim therefore follows from Lemma~\ref{lem:singlemoderelative}.
\end{proof}
\begin{lemma}\label{lem:singlemoderelative}
Let $\vec{\eta}=(\eta_Q,\eta_P)\in\mathbb{R}^2, \gamma\in\mathbb{R}$, and let
$D=\diag(\gamma,\gamma)$ be the covariance matrix of a single mode Gaussian state. Then
\begin{align}
S(\rho_{D, \vec{\eta}}\|\rho_{D,0})=\frac{\beta}{2} (\photonnumber+\eta_Q^2+\eta_P^2)-g(\photonnumber)-\log(1-e^{-\beta})\ , \label{eq:relativeentropyrhorhogamma}
\end{align}
where the mean photon number $\photonnumber$ and the inverse temperature $\beta$ are given by~\eqref{eq:meannumbersymplectic} and~\eqref{eq:inversetemperaturedef}, respectively.
\end{lemma}
\begin{proof}
For brevity, let us write $\rho_D=\rho_{D,0}$ for the centered state. Then we have
\begin{align}
S(\rho_{D,\vec{\eta}}\|\rho_{D})=-S(\rho_D)-\tr(\rho_{D,\vec{\eta}}\log \rho_{D,0})=-g(\photonnumber(\gamma))-\tr(\rho_{D,\vec{\eta}}\log \rho_{D,0})\ .\label{eq:initialexprrelative}
\end{align}
To compute the latter term, we use the expression
$\rho_{D}=\frac{e^{-\beta \hat{n}}}{\tr(e^{-\beta\hat{n}})}=(1-e^{-\beta}) e^{-\beta\hat{n}}$, where $\hat{n}=a^\dagger a=\frac{1}{2}(Q^2+P^2-1)$ is the number operator and $\beta=\beta(\gamma)=\log (\gamma+1)/(\gamma-1)$ the inverse temperature. 
By taking the logarithm, one gets
\begin{align*}
-\tr(\rho_{D,\vec{\eta}}\log \rho_{D})=-\log(1-e^{-\beta})+\beta\tr(\rho_{D,\vec{\eta}}\hat{n})\ .
\end{align*}
Using the fact that $\rho_{D,\vec{\eta}}=W(\vec{\eta})\rho_{D}W(\vec{\eta})^\dagger$ for the Weyl operator $W(\vec{\eta})$ and the fact that 
\begin{align*}
W(\vec{\eta})^\dagger \hat{n} W(\vec{\eta})=\frac{1}{2}((Q+\eta_Q)^2+(P+\eta_P)^2-1)
\end{align*} according to~\eqref{eq:heisenbergactiondisplacement} and Definition~\eqref{eq:numberoperatorsev}, we get
\begin{align*}
-\tr(\rho_{D,\vec{\eta}}\log \rho_{D})&=-\log(1-e^{-\beta})+\frac{\beta}{2}\tr(\rho_{D}((Q+\eta_Q)^2+(P+\eta_P)^2-1/2))\\
&=-\log(1-e^{-\beta})+\frac{\beta}{2}\left(\tr(\rho_D\hat{n})+\eta_Q^2+\eta_P^2\right)\ .
\end{align*}
In the last line, we made use of the fact that $\rho_D$ is centered. The claim follows by combining this with~\eqref{eq:initialexprrelative}. 
\end{proof}
With~\eqref{eq:relativeentropydisplaced}, we can easily compute
the Fisher information of a family of displaced states.
\begin{lemma}[Fisher information of displaced states]\label{lem:fisherdisplaced}
Let~$M_{AB}$ be the covariance matrix of a centered state $\rho_{M_,0}$ of $m+n$ modes, where $S_{AB}M_{AB}S^T_{AB}=D_{m+n}(\vec{\gamma})$.  Fix some $k\in\{1,\ldots,2m\}$ and consider the family of states~$\{\rho^{(\theta,R_{k})}\}_{\theta\in\mathbb{R}}$, 
\begin{align}
\rho^{(\theta,R_k)}=\rho_{M_{AB},\theta \ket{k}}
=\cW_{\theta\ket{k}}(\rho_{M_{AB},0})
\end{align}
 obtained by displacing the state~$\rho_{M,0}$ in the direction $R_{k}$ by an amount~$\theta\in\mathbb{R}$.  Then
\begin{align*}
J(\rho^{(\theta,R_{k})};\theta)|_{\theta=0}=\sum_{j=1}^{m+n}\beta(\gamma_j)  (S_{2j-1,k}^2+S_{2j,k}^2)\ .
\end{align*}
\end{lemma}
\begin{proof}
With~\eqref{eq:relativeentropydisplaced}, we obtain
\begin{align*}
S(\rho_{M,0}\|\rho^{(\theta,R_{k})})=F(\gamma)+\frac{\theta^2}{2}
\sum_{j=1}^{m+n}\beta(\gamma_j)  (S_{2j-1,k}^2+S_{2j,k}^2)\ .
\end{align*}
The Fisher information is the second derivative of this quantity with respect to~$\theta$ at~$\theta=0,$ hence the claim follows.
\end{proof}
\subsection{The de Bruijin identity for conditional entropies of Gaussian states\label{sec:debruijin}}

It will be convenient to define the conditional Fisher information
\begin{align}
J(A|B)_{\rho_{AB}}:=\sum_{k=1}^{2m} J(\rho^{(\theta,R_{k})};\theta)|_{\theta=0}\ \label{eq:Jrhoabdef}
\end{align}
by summing over the modes corresponding to system $A$ only. 
Observe that many properties of the Fisher information carry over to this definition: for example, we have monotonicity
\begin{align}
J(\cE(A)|B)\leq J(A|B)\ ,\label{eq:monotonicityconditional}
\end{align}
for any CPTPM acting on~$\cE$, where these quantities are evaluated on the states $\rho_{AB}$ and $(\cE\otimes I_B)(\rho_{AB})$, respectively. 

By combining Lemma~\ref{lem:entropyincreasediff} and Lemma~\ref{lem:fisherdisplaced}, we obtain a proof of the following statement.
\begin{theorem}[de Bruijin identity for Gaussian states and conditional entropy]\label{thm:debruijingaussiancond}
Let $\rho_{AB}$ be a centered Gaussian state of $m+n$ modes.  Define~$\rho_{AB}(t)=(e^{t\cL_A}\otimes I_B)(\rho_{AB})$.  Then
\begin{align*}
\frac{d}{dt}\Big|_{t=0}S(A|B)_{\rho_{AB}(t)}=\frac{1}{4}J(A|B)_{\rho_{AB}}\ .
\end{align*}
\end{theorem}

\begin{proof}
It is straightforward to check that~$S(I_A\oplus 0_B) S^T$ has
diagonal elements of the form $(S(I_A\oplus 0_B) S^T)_{\ell,\ell}=\sum_{k=1}^{2m} S_{\ell,k}^2$.  Hence
\begin{align*}
\frac{d}{dt}\Big|_{t=0}S\left((e^{t\cL_A}\otimes I_B)(\rho_{AB})\right)&=\frac{1}{4}
\sum_{j=1}^{m+n}\beta(\gamma_j)\left( (S(I_A\oplus 0_B) S^T)_{2j-1,2j-1}+(S(I_A\oplus 0_B) S^T)_{2j,2j}\right)\\
&=\frac{1}{4}
\sum_{j=1}^{m+n}\sum_{k=1}^{2m}\beta(\gamma_j)\left( S^2_{2j-1,k}+S^2_{2j,k}\right)\ 
\end{align*}
according to Lemma~\ref{lem:entropyincreasediff}. We conclude from Lemma~\ref{lem:fisherdisplaced} and  Definition~\eqref{eq:Jrhoabdef} that
\begin{align*}
\frac{d}{dt}\Big|_{t=0}S\left((e^{t\cL_A}\otimes I_B)(\rho_{AB})\right)=\frac{1}{4} J(A|B)_{\rho_{AB}}
\end{align*}
The claim then follows because~$\rho_{B}(t)=\rho_B$ does not depend on time (cf.~\eqref{eq:evolvedsabt}). 
\end{proof}

\section{The entropy power inequality for conditional entropy\label{sec:conditionalentropypower}}
Having established the de Bruijin identity for conditional entropies as well as the asymptotic scaling of the conditional entropies under diffusion, it is straightforward to prove the entropy power inequality for conditional entropies. Indeed, this follows the pattern of known classical proofs~\cite{Demboetal91}, with minor modifications because  we are considering conditional entropies. It relies heavily on the Fisher information inequality (a consequence of data processing, as shown by Zamir~\cite{Zamir98}). 

We introduce the necessary definitions in Section~\ref{sec:beamsplittersproductstates}. The conditional Fisher information inequality and the conditional entropy power inequality are derived subsequently in Sections~\ref{sec:fisherinfo} and~\ref{sec:condtionalentropypowerga}.

\subsection{Beam splitters, product states and auxiliary  systems\label{sec:beamsplittersproductstates}}
Consider two systems $X_j$,  $j=1,2$ with $N$~modes each and associated mode operators~$\{Q_k^{(j)},P_k^{(j)}\}_{k=1}^N$.
A beam-splitter with transmissivity~$\lambda\in [0,1]$ acting on $X_1X_2$ is the Gaussian unitary $U_{S_\lambda}$ described by the symplectic 
matrix
\begin{align}
S_{\lambda} &=
\left(
\begin{matrix}
\sqrt{\lambda} I_{2N} & \sqrt{1-\lambda} I_{2N}\\
\sqrt{1-\lambda} I_{2N} & -\sqrt{\lambda} I_{2N}\ 
\end{matrix}
\right)\label{eq:Slambdadef}
\end{align}
with respect to the ordering $(Q_1^{(1)},P_1^{(1)},\ldots,Q_N^{(1)},P_N^{(1)}, Q_1^{(2)},P_1^{(2)},\ldots,Q_N^{(2)},P_N^{(2)})$ of modes. We are interested in the beam-splitter map
$\cE_\lambda=\cE_\lambda^{X_1X_2\mapsto Y}$, which 
is obtained by letting~$X_1,X_2$ interact according to $U_{S_\lambda}$, and discarding the second set of $N$ modes.  That is, it is a map from $2N$~input modes to 
$N$~output modes; we call the latter~$Y$. Formally, the map~$\cE_\lambda$ is defined as
\begin{align*}
\cE_{\lambda}(\rho_{X_1X_1})&=\tr_{X_2}\left(U_{\lambda}\rho_{X_1X_2}U_\lambda^\dagger\right)\ ,
\end{align*}
where $\tr_{X_2}$ denotes the partial trace of over the second set~$X_2$ of modes.  We will denote the output system (i.e., the set of modes~$X_1$ at the end of this process) by~$Y$, i.e., think of $\cE_\lambda=\cE_\lambda^{X_1X_2\mapsto Y}$ as a map from systems~$X_1X_2$ to an output system~$Y$ (of~$N$ modes). Since the partial trace~$\tr_{X_2}$ is a Gaussian map, the map~$\cE_\lambda$ is Gaussian and completely determined by its action on  covariance matrices and displacement vectors. This is
\begin{align*}
\begin{matrix}
\rho_{X_1X_2}&\mapsto \cE_\lambda(\rho_{X_1X_2})\\
\left(\left(\begin{matrix}
M_{X_1} &L_{X_1X_2}\\
L_{X_1X_2}^T & M_{X_2}
\end{matrix}\right),(\vec{d}_{X_1},\vec{d}_{X_2})
\right)&\overset{\cE_\lambda}{\mapsto}
(M_{Y},\vec{d}_Y)
\end{matrix}\end{align*}
where
\begin{align}
\begin{matrix}
M_Y&=&\lambda M_{X_1}+(1-\lambda) M_{X_2}+\sqrt{\lambda(1-\lambda)}(L_{X_1X_2}+L_{X_1X_2}^T)\\
\vec{d}_Y&=&\sqrt{\lambda }\vec{d}_{X_1}+\sqrt{1-\lambda} \vec{d}_{X_2}\  .
\end{matrix}\label{eq:transformationelambdacov}
\end{align}
This follows immediately from~\eqref{eq:Slambdadef}. 

\begin{center}
\begin{figure}
\begin{center}
\includegraphics{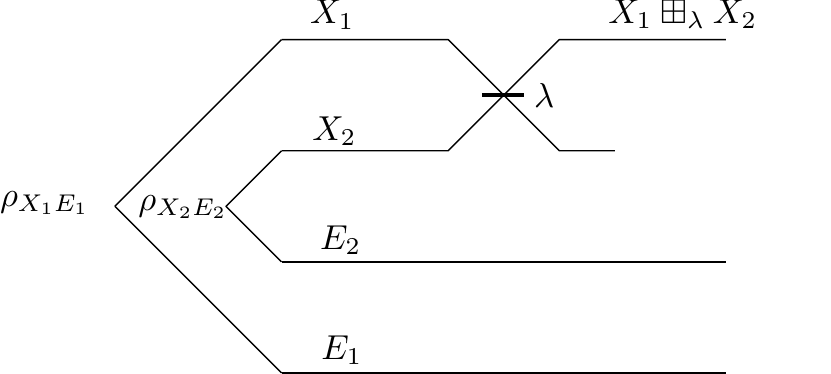}
\end{center}
\caption{\label{fig:beamsplitterdef}This quantum circuit circuit defines the states~$\rho_{X_1\boxplus_\lambda X_2}$ and  $\rho_{(X_1\boxplus_\lambda X_2)E_1E_2}$
 for any product input
$\rho_{X_1E_1}\otimes\rho_{X_2E_2}$.}
\end{figure}
\end{center}
We will consider input states that are products (across the bipartition $X_1:X_2$). It will be convenient to introduce the following maps and states as summarized in Figure~\ref{fig:beamsplitterdef}:
\begin{description}
\item[the state $\rho_{X_1\boxplus_\lambda X_2}$, given $\rho_{X_1}, \rho_{X_2}$: ] For product inputs $\rho_{X_1}\otimes\rho_{X_2}$, the transformation~\eqref{eq:transformationelambdacov}
specializes to 
\begin{align*}
\begin{matrix}
\rho_{X_1}\otimes\rho_{X_2} & \mapsto& \sigma_Y:= \cE^{X_1X_2\mapsto Y}_\lambda(\rho_{X_1}\otimes\rho_{X_2})\\
\left(\left(\begin{matrix}
M_{X_1} &0 \\
0 & M_{X_2}
\end{matrix}\right),(d_{X_1},d_{X_2})
\right) &\overset{\cE_\lambda}{\mapsto}&
(\lambda M_{X_1}+(1-\lambda) M_{X_2},\sqrt{\lambda}\vec{d}_{X_1}+\sqrt{1-\lambda}\vec{d}_{X_2})
\end{matrix}\ ,
\end{align*}
where we assume that $\rho_{X_j}$ is a Gaussian state described by~$(M_{X_j},\vec{d}_{X_j})$, $j=1,2$. We will denote the output state~$\sigma_Y$ obtained in this fashion by $\rho_{X_1\boxplus_\lambda X_2}$.

\item[the state $\rho_{(X_1\boxplus_\lambda X_2)E_1E_2}$ given $\rho_{X_1E_1}$ and $\rho_{X_2E_2}$:]
More generally, consider the map
$I_{E_1E_2}\otimes \cE_\lambda^{X_1X_2\mapsto Y}$, where  $E_j$, $j=1,2$ are auxiliary 
 systems with~$L$ modes each. We will only need a description of  its action on product states $\rho_{X_1E_1}\otimes\rho_{X_2E_2}$, where $\rho_{X_jE_j}$ are Gaussian states with covariance matrix and displacement vectors
\begin{align}
M_{X_jE_j}&=
\left(
\begin{matrix}
M_{X_j} & L_{X_jE_j}\\
L_{X_jE_j}^T & M_{E_j}
\end{matrix}
\right)\ \qquad \vec{d}_j=(\vec{d}_{X_j},\vec{d}_{E_j})\qquad\textrm{ for }j=1,2\ .\label{eq:twostatecovariancematrices}
\end{align}
It is straightforward to verify that this is given by
\begin{align*}
\rho_{X_1E_1}\otimes\rho_{X_2E_2}&\mapsto \sigma_{YE_1E_2}:=(I_{E_1E_2}\otimes\cE^{X_1X_2\mapsto Y}_\lambda) (\rho_{X_1E_1}\otimes\rho_{X_2E_2})\\
(M_{X_1E_1}\oplus M_{X_2E_2},(\vec{d}_1,\vec{d}_2))&\mapsto (M_{YE_1E_2},(\vec{d}_Y,\vec{d}_{E_1},\vec{d}_{E_2}))
\end{align*}
where
\begin{align}
\begin{matrix}
M_{YE_1E_2}&=&\left(
\begin{matrix}
\lambda M_{X_1}+(1-\lambda) M_{X_2} & \sqrt{\lambda} L_{X_1E_1} & \sqrt{1-\lambda}L_{X_2E_2}\\
 \sqrt{\lambda} L_{X_1E_1}^T & M_{E_1} & 0\\
\sqrt{1-\lambda}L_{X_2E_2}^T & 0 & M_{E_2}
\end{matrix}
\right)\ \\
 \vec{d}_Y&=&\sqrt{\lambda} \vec{d}_{X_1}+\sqrt{1-\lambda}\vec{d}_{X_2}\ .
\end{matrix}\label{eq:transformationrulejointmap}
\end{align}
\end{description}

\subsection{The conditional Fisher information inequality for beamsplitters\label{sec:fisherinfo}}
In~\cite[Lemmas 3.2 and 6.1]{KoeGrae}, it was shown that
the beam-splitter map is compatible with both diffusion and translations in the following sense. For all $t\geq 0$, ~$w_1,w_2\in\mathbb{R}$, $\vec{\theta}\in\mathbb{R}^{2m}$, we have the following identitities of  Gaussian maps:
\begin{align}
\cE_\lambda^{X_1X_2\mapsto Y}\circ (e^{t\cL_{X_1}}\otimes e^{t\cL_{X_2}})&=e^{t\cL_Y}\circ\cE_\lambda^{X_1X_2}\label{eq:diffusioncompatibility}\\
\cE_\lambda^{X_1X_2\mapsto Y}\circ (\cW^{X_1}_{w_1\vec{\theta}}\otimes 
\cW^{X_2}_{w_2\vec{\theta}})&=\cW^{Y}_{w\vec{\theta}}\circ \cE_\lambda^{X_1X_2\mapsto Y}\qquad\textrm{ where }w=\sqrt{\lambda}w_1+\sqrt{1-\lambda}w_2\ .\label{eq:translationcompatibility}
\end{align}
This was then used to show that the quantity~$J$ (cf.~\eqref{eq:Jrhouncond}) satisfies the
Fisher information inequality
\begin{align}
J(X_1\boxplus_\lambda X_2)\leq \lambda J(X_1)+(1-\lambda)J(X_2)\ .\label{eq:fisherinformationinequality}
\end{align}
The proof of~\eqref{eq:fisherinformationinequality} follows immediately from the monotonicity~\eqref{eq:fisherinformationmotonicity}  of the divergence-based Fisher information, its additivity~\eqref{eq:fisherinformationadditivity},  the reparametrization identity~\eqref{eq:reparametrizationidentity}, as well as the compatibility properties~\eqref{eq:diffusioncompatibility},~\eqref{eq:translationcompatibility}. We will omit the corresponding argument here; it was discovered in the classical context by Zamir~\cite{Zamir98}.

Here we argue briefly that the quantity~\eqref{eq:Jrhoabdef} satisfies an analogous inequality, that is,
\begin{align}
J(X_1\boxplus_\lambda X_2|E_1E_2)\leq \lambda J(X_1|E_1)+(1-\lambda)J(X_2|E_2)\ .\label{eq:fisherinformationinequality} 
\end{align}
Indeed, it is clear that the identities~\eqref{eq:diffusioncompatibility} and 
\eqref{eq:translationcompatibility} still hold if we replace
the maps $e^{t\cL_{Z}}$, $\cE_\lambda^{X_1X_2\mapsto Y}$ and $\cW^{Z}_{\vec{\xi}}$
by their `stabilized' versions (obtained by adjoining an identity)
\begin{align*}
\begin{matrix}
e^{t\cL}\qquad & \mapsto &\qquad e^{t\cL}\otimes I_{E_1E_2}\\
\cE_\lambda^{X_1X_2\mapsto Y}\qquad & \mapsto &\qquad \cE_\lambda^{X_1X_2\mapsto Y}\otimes I_{E_1E_2}\\
\cW^{Z}_{\vec{\xi}} & \mapsto & \cW^{Z}_{\vec{\xi}}\otimes I_{E_1E_2}\ .
\end{matrix}
\end{align*}
Furthermore, 
the quantity $J(A|B)_{\rho_{AB}}$ is also motononous (cf.~\eqref{eq:monotonicityconditional}). Because
each of the terms~$J(\rho^{\theta,R_k};\theta)|_{\theta=0}$  constituting~$J(A|B)_{\rho_{AB}}$ is additive and satisfies the reparametrization identities, we can apply Zamir's proof again (carrying along~$E_1E_2$) and obtain
the conditional Fisher information inequality~\eqref{eq:fisherinformationinequality}.

\subsection{The conditional entropy power inequality for Gaussian states\label{sec:condtionalentropypowerga}}
The entropy power inequality we prove relates the 
conditional entropy $S(Y|E_1E_2)_\sigma$ of the state
\begin{align*}
\sigma_{YE_1E_2}&=(\cE_\lambda\otimes I_{E_1E_2})(\rho_{X_1E_1}\otimes\rho_{X_2E_2}):=\rho_{(X_1\boxplus_\lambda X_2)E_1E_2}
\end{align*} 
to the conditional entropies $S(X_j|E_j)$ of the two (Gaussian) states $\rho_{X_jE_j}$, $j=1,2$. 

\begin{theorem}[Conditional entropy power inequality for Gaussian states]\label{thm:entropypowerinequalitygaussian}
Let $\rho_{X_1E_1}$ and $\rho_{X_2E_2}$ be arbitrary  Gaussian states.
Then 
\begin{align*}
S(X_1\boxplus_\lambda X_2|E_1E_2)\geq \lambda S(X_1|E_1)+(1-\lambda)S(X_2|E_2)\ .
\end{align*}
\end{theorem}
Note that the proof outlined here combined with the discussion in~\cite{KoeGrae} (respectively Stam's proof~\cite{Stam59}) should also provide the inequality
\begin{align*}
e^{S(X_1\boxplus_\lambda X_2|E_1E_2)}\geq \frac{1}{2}e^{S(X_1|E_1)/n}+\frac{1}{2}e^{S(X_2|E_2)/n}
\end{align*}
where $X_1$ and $X_2$ have $n$~modes. We do not discuss this version here for brevity.

\begin{proof}
Let the covariance matrices and displacement vectors of
$\rho_{X_jE_j}$,  $j=1,2$, be given by~\eqref{eq:twostatecovariancematrices}.
The corresponding covariance matrices are
\begin{align*}
M_{X_j(t)E_j}&=
\left(
\begin{matrix}
M_{X_j(t)} & L_{X_jE_j}\\
L_{X_jE_j}^T & M_{E_j}
\end{matrix}
\right)\qquad\textrm{ where }M_{X_j(t)}=M_{X_j}+tI\ .
\end{align*}
For $t\geq 0$, define the function
\begin{align*}
\delta(t):=S(X_1(t)\boxplus_\lambda X_2(t)|E_1E_2)-\lambda S(X_1(t)|E_1)-(1-\lambda)S(X_2(t)|E_2)\ ,
\end{align*}
where the entropies are evaluated on the states~$\rho_{X_j(t)E_j}$, $j=1,2$ and
the result~$\rho_{(X_1(t)\boxplus_\lambda X_2(t))E_1E_2}$ of letting these interact with the beamsplitter (as discussed in Section~\ref{sec:beamsplittersproductstates}).
According to the `stabilized' version of~\eqref{eq:diffusioncompatibility}, we have
\begin{align*}
\delta(t)&=S((X_1\boxplus X_2)(t)|E_1E_2)-\lambda S(X_1(t)|E_1)-(1-\lambda)S(X_2(t)|E_2)\\
&=S(Y(t)|E_1E_2)-\lambda S(X_1(t)|E_1)-(1-\lambda)S(X_2(t)|E_2)\ .
\end{align*}
This shows that $\delta(t)$ is the difference of conditional entropies 
of time-evolved states for different initial states~$\rho_{YE_1E_2}$, $\rho_{X_1E_1}$ and $\rho_{X_2E_2}$ at $t=0$. We conclude with Lemma~\ref{lem:conditionalentropyinfinitescaling}
that
\begin{align}
\lim_{t\rightarrow\infty}\delta(t)=0\ .\label{eq:deltalimit}
\end{align}
On the other hand, we have according to the de Bruijin identity (Theorem~\ref{thm:debruijingaussiancond})
\begin{align*}
4\delta'(t)&=J(X_1(t)\boxplus_\lambda X_2(t)|E_1E_2)-\lambda J(X_1(t)|E_1)-(1-\lambda)J(X_2(t)|E_2)\ .
\end{align*}
This identity, together with Fisher information inequality~\eqref{eq:fisherinformationinequality}  imply that $\delta'(t)\leq 0 $ for all $t\geq 0$. With~\eqref{eq:deltalimit}, this shows that $\delta(0)\geq 0$, which is the claim. 
\end{proof}
\subsubsection*{Acknowledgements}
I would like to thank the organizers of the workshop `Beyond iid in quantum information theory'. I also thank Reinhard Werner, Graeme Smith and Jon Yard for discussions, and gratefully acknowledge support by NSERC.
\bibliographystyle{plain}
\bibliography{q}
\end{document}